\documentclass[12pt, draftclsnofoot, onecolumn]{IEEEtran}
\usepackage{amsthm}
\usepackage{amsmath}
\usepackage{algorithmic}
\usepackage[ruled]{algorithm2e}
\usepackage{mathrsfs}
\usepackage{graphicx}
\usepackage{subfigure}
\usepackage{cite}
\usepackage{bm,epsfig,amsthm,url}
\usepackage{indentfirst}
\usepackage{amssymb}
\usepackage{epstopdf}
\usepackage{xcolor}
\usepackage{mathtools}

\theoremstyle{definition}
\newtheorem{assumption}{Assumption}

\newtheorem{theorem}{Theorem}

\newtheorem{proposition}{Proposition}

\newtheorem{remark}{Remark}

\begin{document}

\title{Federated Learning via Unmanned Aerial Vehicle}
\author{Min~Fu,~\IEEEmembership{Member,~IEEE},  Yuanming~Shi,~\IEEEmembership{Senior Member,~IEEE}, 
and~Yong~Zhou,~\IEEEmembership{Member,~IEEE} 
\thanks{ M. Fu is with the Department of Electrical and Computer Engineering, National University of Singapore, Singapore 117583 (e-mail: fumin@u.nus.edu). }
\thanks{Y. Shi and Y. Zhou  are with School of Information Science and Technology, ShanghaiTech University, Shanghai 201210, China (e-mail: \{shiym, zhouyong\}@shanghaitech.edu.cn).} 
}
\maketitle
\setlength\abovedisplayskip{2pt}
\setlength\belowdisplayskip{2pt}
\setlength\abovedisplayshortskip{2pt}
\setlength\belowdisplayshortskip{2pt}

\begin{abstract}
	To enable communication-efficient federated learning (FL), this paper studies an unmanned aerial vehicle (UAV)-enabled FL system, where the UAV coordinates distributed ground devices for a shared model training.
Specifically, by exploiting the UAV's high altitude and mobility, the UAV can proactively establish short-distance line-of-sight links with devices and prevent any device from being a communication straggler.
Thus, the model aggregation process can be accelerated while the cumulative model loss caused by device scheduling can be reduced, resulting in a decreased completion time.
We first present the convergence analysis of FL without the assumption of convexity, demonstrating the effect of device scheduling on the global gradients.
Based on the derived convergence bound, we further formulate the completion time minimization problem by jointly optimizing device scheduling, UAV trajectory, and time allocation. 
This problem explicitly incorporates the devices' energy budgets, dynamic channel conditions, and convergence accuracy of FL constraints.
Despite the non-convexity of the formulated problem, we exploit its structure to decompose it into two sub-problems and further derive the optimal solutions via the Lagrange dual ascent method.
Simulation results show that the proposed design significantly improves the tradeoff between completion time and prediction accuracy in practical FL settings compared to existing benchmarks.
\end{abstract}

\begin{IEEEkeywords}
	Federated leaning, UAV communications, completion time minimization, device scheduling, and UAV trajectory design.
\end{IEEEkeywords}

\section{Introduction}
As the storage and computation capabilities of  edge devices keep growing, it becomes more attractive to process the data locally and push network computation to the edge \cite{Lim2020Edge, Shi2020edgeAI, Letaief2022Edge}.
In the field of machine learning (ML), distributed learning frameworks \cite{Li2020FLChallenges, Zhou2019Edge} that keep the training data locally are well developed to protect data privacy and reduce network energy/time costs.
Recently, federated learning (FL) \cite{Li2020FLChallenges, Zhou2019Edge, Lim2020Edge} has  been proposed as a promising solution for distributed ML, which enables multiple devices to execute local training on their own dataset and collaboratively build a shared ML model with the coordination of a parameter server (PS) (e.g., access point and base station).
Since only model parameters rather than raw data are exchanged between devices and the PS, FL significantly relieves the communication burden and protects data privacy\cite{Li2020FLChallenges,Yang2020Greedy} with wide-field applications, e.g., vehicle-to-vehicle communications \cite{Savazzi2021Opportunities} and content recommendations for smartphones\cite{Zhou2019Edge}.

In contrast to the centralized ML, the PS in FL needs to exchange models with multiple devices over hundreds to thousands of communication rounds to achieve the desired training accuracy.
However, the main challenge in realizing FL on wireless networks arises from communication stragglers with unfavorable links \cite{Li2021Delay, Imteaj2022Survey}.
For example, in over-the-air computation (AirComp)-based analog FL \cite{Yang2020Greedy, Zhu2020Broadband}, communication stragglers dominate the overall model aggregation error caused by channel fading and communication noise since the devices with better channel qualities have to reduce their transmit power for the local models' alignment at the PS.
Moreover, in digital synchronous FL \cite{Chen2021FL_TWC, Li2021Delay}, communication stragglers significantly slow down the model aggregation process and dominate cumulative communication delay since the PS must wait until receiving the training updates from all participants.
If the number of communication stragglers is high, the overall communication delay will be unacceptable.
The straggler issue is thus the main bottleneck to design communication-efficient FL systems.

There have been many efforts to mitigate the communication straggler effect in FL, such as device scheduling \cite{Yang2020Greedy, Zhu2020Broadband, Chen2021FL_TWC}.
For instance, to reduce model misalignment error incurred by stragglers in AirComp-based FL, the authors in \cite{Yang2020Greedy, Zhu2020Broadband} scheduled the devices with reliable channels for concurrent model uploading.
In addition, to reduce the communication delay incurred by stragglers in digital FL, devices with large contributions to the global model \cite{Chen2021FL_TWC} or/and with favorable channel conditions \cite{Ren2020Scheduling, Amiria2021FL_TWC} are generally selected.
Nevertheless, because such device scheduling is biased, which results in a smaller amount of training data utilized, this, in turn, may damage the update of the global model and decrease the learning performance of FL.
To alleviate such communication-learning tradeoff, recent research has investigated the integration of advanced technologies (i.e., relays \cite{Lin2022RelayFL}, reconfigurable intelligent surfaces \cite{Liu2021RISGreedy, Wang2022IRS, Yang2022differentially}) into FL systems to improve stragglers' communication qualities and thus further upgrade device scheduling policy for the reduction of communication errors.
These existing frameworks require a terrestrial BS to provide network coverage to the devices for model aggregation. 
However, many FL tasks need to be performed under the circumstances when terrestrial networks are unavailable in remote areas.
For example, devices from multiple regions (e.g., forests and woodlands) can collaborate through FL to build a learning model for fire monitoring \cite{Chen2021Satellite}.
Under these harsh environments, it is imperative to deploy a more flexible PS that proactively establishes favorable communication links among devices.

As a viable complementary alternative to terrestrial networks, unmanned aerial vehicles (UAVs) can provide coverage extension and seamless connectivity to support various FL tasks, especially in distant and underdeveloped areas \cite{Zeng2019Accessing, Lim2021UAV, Fu2021UAV}.
Inspired by this, this paper studies a UAV-enabled FL network, where a UAV is dispatched as a flying PS to aggregate and update the digital FL model parameters when no terrestrial BS is available.
To mitigate the communication straggler effect in UAV-enabled FL networks, we propose to jointly design UAV trajectory and device scheduling.
First, the qualities of communication channels between the UAV and devices still differ since all links’ channel conditions depend on the UAV’s location at each time slot.
To address this issue, device scheduling is necessary to prevent communication stragglers.
Furthermore, by utilizing its mobility, the UAV can establish short-distance LoS links to scheduled devices.
As a result, each scheduled device achieves a high data rate for model uploading, resulting in faster model uploads/downloads per round compared to the static UAV.
In addition, with its ability to dynamically adjust communication distances, the UAV can prevent any device from being a communication straggler all the time, ensuring that all devices have the opportunity to participate in FL training. 
As such, incorporated with UAV trajectory planning, the optimized device scheduling strategy focuses on data exploitation maximization, thereby reducing cumulative model aggregation loss and accelerating FL convergence.
Hence, such join design in UAV-enabled FL networks is an appealing solution to address the straggler issue and reduce communication time.

Motivated by the above observations, in this paper, we focus on the time required for completing the FL training process \cite{Dinh2021FL_TON,  Ma2021FedSA}, which includes not only the computation time but also the communication time of all scheduled devices.
The completion time is a critical design aspect in UAV-enabled systems, as UAVs usually have a limited endurance due to the practical physical constraints \cite{Zhan2019Time}.
Additionally, devices generally have a limited energy budget without battery replacement or recharging \cite{Imteaj2022Survey}.
Thus, the formulated optimization problem needs to consider not only the FL constraints such as the target convergence accuracy level but also the devices' dynamic communication conditions such as energy consumption.
As such, this paper aims to minimize the completion time of UAV-enabled FL by jointly optimizing the UAV trajectory, device scheduling, and time allocation for energy-constrained devices, while ensuring that the FL algorithm can converge to a target accuracy.

\subsection{Contributions}
The main contributions of this paper are summarized as follows.
\begin{itemize}
	\item 
	We consider a novel UAV-enabled FL framework, where a mobile UAV is dispatched as a PS to exchange the model parameters with devices, thereby preventing devices from being communication stragglers.
	Moreover, under the convergence accuracy of FL constraint, we mainly focus on the completion time minimization problem for UAV-enabled FL systems. 
	This consideration is of paramount importance for latency-critical FL applications and limited endurance UAV systems.
	
	\item  
	We first derive an upper bound on the iterative norm of the global gradients for non-convex loss functions,  taking model update errors from device selection into account.
	Specifically, the analytical bound shows that the device selection loss causes convergence rate reduction and leads to a non-diminishing gap between the initial model and the global optimum of the training loss.

	\item Based on the convergence bound, we formulate the completion time minimization problem by jointly designing device scheduling, communication time allocation for scheduled devices, and trajectory planning for UAV, taking into account the target convergence accuracy of FL, energy budget at devices, and practical constraints at the UAV. Besides, we quantitatively analyze the fundamental tradeoff between learning accuracy and communication latency, and demonstrate the importance of the mobile UAV in the proposed system.

	\item Although the formulated mixed-integer nonconvex problem is high intractable, we propose an efficient iterative block coordinate descent (BCD) algorithm to solve the joint device scheduling and time allocation optimization subproblem and the UAV trajectory optimization subproblem alternately, which is guaranteed to converge.
	Moreover, to reduce computational complexity, each optimization subproblem is solved optimally with a closed-form solution by applying the Lagrange duality method.
	
\end{itemize}

Simulation results in realistic federated settings are provided to illustrate the learning-communication tradeoff and validate the effectiveness of the proposed scheme. 
Specifically, the proposed joint design scheme outperforms existing benchmark schemes in terms of mission completion time, which is appealing in light of the limited endurance of UAVs.
Moreover, our proposed scheme provides comparable performance to the full-scheduling ideal benchmark in terms of prediction accuracy, even when the target convergence accuracy of  FL is relatively large.

\vspace{-2mm}
\subsection{Related Work}

\subsubsection{UAV Assisted FL Networks}
Upon the completion of this work, the application of UAV in FL networks was investigated in some parallel works \cite{Liu2021FLUAV, Lim2021FLUAV, Ng2021UAVFL, Yang2021UAVFL}
Specifically, the authors in \cite{Liu2021FLUAV} developed a novel FL framework with UAV swarms to improve the FL learning efficiency.
The authors in \cite{Lim2021FLUAV} proposed an FL-based sensing and collaborative learning approach for UAV-enabled internet of vehicles (IoVs), where UAVs as devices collect data and train ML models for IoVs.
In addition, \cite{Ng2021UAVFL} described using UAVs as flight relays to support wireless communication between IoVs and the FL server, thus enhancing FL accuracy.
The authors in \cite{Yang2021UAVFL} studied the deployment of multiple UAVs as flying BSs to minimize the weighted sum of FL execution time and function loss.
Nevertheless, this work does not consider device budget issues that may affect FL performance and convergence cannot be guaranteed.

\subsubsection{Latency Minimization Problems in FL Networks}
 A few works have studied the completion time minimization problems in different FL scenarios.
 The authors in \cite{Dinh2021FL_TON} formulated an FL framework over a wireless network as an optimization problem that minimizes the sum of FL aggregation latency and total device energy consumption.
 In addition, the authors in \cite{Yang2021Energy} investigated the tradeoff between the FL convergence time and devices’ energy consumption.
 However, in \cite{Dinh2021FL_TON, Yang2021Energy}, all devices are assumed to be involved in each round.

\subsection{Organization}
The remainder of this paper is organized as follows. Section \ref{Section:model} describes the FL via the UAV system model. 
In Section \ref{Section:Convergence}, we provide the convergence analysis of FL and completion time minimization problem formulation.
In Section \ref{Section:algorithm}, we propose a BCD method to solve the formulated problem.
Section \ref{Section:simulation} presents the numerical results to evaluate the performance of the proposed algorithm. 
Finally, we conclude this paper in Section \ref{Section:conclusion}.

\textit{Notations}: In this paper, scalars, column vectors and matrices are written in italic letters, boldfaced lower-case letters and boldfaced upper-case letters respectively, e.g., $a$, $\mathbf{a}$, $\mathbf{A}$. $\mathbb{R}^{M \times N}$ denotes the space of a real-valued matrix with $M$ rows and $N$ columns. 
$\left\| \mathbf{a} \right\|_2$ denotes the Euclidean norm of vector $\mathbf{a}$ and $\mathbf{a}^T$ represents its transpose. $|\mathcal{S}|$ denotes the cardinality of the set $\mathcal{S}$. $\langle \cdot, \cdot \rangle$ represents the inner product.

\section{System Model}\label{Section:model}

We consider a UAV-assisted FL network that consists of a single-antenna UAV and a set $\mathcal{K}$ of $K$ single-antenna devices, aiming to collaboratively learn an ML model (e.g., logistic regression and linear regression). 
Since the terrestrial BSs are usually sparsely deployed or unavailable in rural, remote, and underdeveloped areas, we deploy an UAV with a powerful processor to act as a PS in these harsh circumstances.

Each device $k$ has a local dataset $\mathcal{D}_{k}$ with size $D_k$. 
Local dataset $\mathcal{D}_{k}$ is a collection of data samples $\{\bm x_{ki},y_{ki}\}_{i=1}^{D_k}$, where $\bm x_{ki}\in \mathbb{R}^{d}$ is the input sample vector with $d$ features and $y_{ki}\in \mathbb{R}$ is the label for sample $\bm x_{ki}$.
The UAV-enabled FL system aims to cooperatively learn an ML model under the coordination of the UAV in an iterative manner while keeping the local training data on mobile devices.
In the sequel, we describe FL over UAV-enabled wireless networks, as summarized in Algorithm \ref{Algo:FL}.
\setlength{\textfloatsep}{0.2cm}
\setlength{\floatsep}{0.2cm}
\begin{algorithm}[t] \caption{FL over UAV-enabled wireless networks}\label{Algo:FL}
	\begin{algorithmic}[1]
		\STATE {\textbf{Input:} } $ \eta$, $\bm {w}_0 = \bm 0$, $n = 0$;
		\REPEAT
		\STATE Update $n = n+1$.
		\STATE {\textbf{Device Scheduling}}: Denote the set of scheduled devices $k$ in the $n$-th round as $\mathcal{A}_n$.
		\STATE {\textbf{Computation}}: All devices in set $\mathcal{A}_n$ receive $\bm {w}_{n-1}$ from the server, and update their local FL models denoted as $\bm {w}_{k,n}, k\in \mathcal{A}_n$ according to \eqref{Eq:local update} in parallel.
		\STATE {\textbf{Communication}}: All the scheduled devices $k\in \mathcal{A}_n$ transmit $\bm {w}_{k,n}$ to the UAV over TDMA.
		\STATE Aggregation and Feedback: The UAV updates the global model $\bm {w}_{n}$ as in \eqref{Eq:global update}, and then feed-back the updated global model to devices.
		\UNTIL{$n = N$}
		\STATE {\textbf{Output:} the global model $\bm {w}_{N}$}	 		
	\end{algorithmic}
\end{algorithm}

\subsection{FL Training Model}
Let $\bm w\in \mathbb{R}^{d}$ and $\bm w_k\in \mathbb{R}^{d}$ be the global model parameters at the UAV and the local model parameters at the $k$-th device, respectively.
The local loss function of device $k$ is defined as
\setlength\arraycolsep{2pt}
\begin{eqnarray} \label{Eq:local loss}
	\mathcal{F}_k(\bm{w}) = \frac{1}{D_k}\sum_{i = 1}^{D_k}f(\bm{w}; \bm x_{ki},y_{ki}), \forall k \in \mathcal{K}.
\end{eqnarray}
Accordingly, the global loss function at the UAV is given by
\begin{eqnarray}\label{Eq:global loss}
	\mathcal{F}(\bm w):= \frac{1}{D}\sum_{k =1}^KD_k\mathcal{F}_k(\bm{w}),
\end{eqnarray}
where $D$ denotes the total size of data with $D = \sum_{k=1}^K D_k$.
The training of an FL algorithm aims to solve the following optimization problem:
\begin{eqnarray}\label{Problem:FL minimizer}
	\mathop{\text{minimize}}_{\bm w_1,\ldots, \bm w_K}&& \frac{1}{D}\sum_{k =1}^K\sum_{i =1}^{D_k} f(\bm{w}_k; \bm x_{ki}, y_{ki})\nonumber\\
	\text{subject to} 
	&& \bm w_1 =\cdots=\bm w_K = \bm w, \label{Cons:weight}
\end{eqnarray}
where constraint \eqref{Cons:weight} is used to ensure that all devices and the server have the same model after the FL algorithm converges.
To solve problem \eqref{Problem:FL minimizer}, the training procedure of FL consists of multiple communication rounds.
Therein, in a communication round, the devices update the local models based on their own data and a common model initialization from the PS. 
After that, they contribute to the global model at the PS by uploading model updates over the wireless channel, as described below.

First, the UAV determines the set of devices that participate in the current round.
Let $a_k[n]\in\{0,1\}$ denote the scheduling variable for device $k$ in the $n$-th round.
Therein, $a_k[n] =1$ indicates that device $k$ sends its updated local model to the UAV at communication round $n$; otherwise we have $a_k[n] =0$.
$\mathcal{A}_n$ is defined as the set of scheduled devices in the $n$-th round.
Then, the UAV broadcasts the current global model, denoted by $\bm w_{n-1}$, to all the scheduled devices.

Each scheduled device $k \in \mathcal{A}_n$ receives the global model and updates its local model by applying the gradient descent algorithm \cite{Chen2021FL_TWC, Li2021Delay} on its local dataset:
\begin{eqnarray}\label{Eq:local update}
	&&\bm w_{k,n} = \bm w_{n-1} - \frac{\eta}{D_k}\sum_{(x_{ki}, y_{ki})\in \mathcal{D}_{k}}\bigtriangledown f(\bm w_{n-1};\bm x_{ki}, y_{ki}),
\end{eqnarray}
where $\eta$ is the learning rate and $\bigtriangledown f( \bm w_{n-1};\bm x_{ki}, y_{ki})$ is the gradient of $f( \bm w;\bm x_{ki}, y_{ki})$ with respect to $\bm w$ at $\bm w_{n-1}$.

After all scheduled devices upload their local models, the UAV aggregates them to obtain the following consolidated global model
\begin{eqnarray}\label{Eq:global update}
	&&	\bm w_n = \frac{\sum_{k=1}^Ka_k[n]\bm w_{k,n}}{\sum_{k=1}^Ka_k[n]},
\end{eqnarray}
where $\bm w_n$ is the global model updated at the UAV in the $n$-th round.
With \eqref{Eq:local update} and \eqref{Eq:global update}, the UAV and the devices can repeatedly update their models until the maximum number of communication rounds reaches. 
As the local FL models are updated at the scheduled devices and transmitted over wireless cellular links, device scheduling determines both the learning performance and communication delay.
Furthermore, the channel conditions between the UAV and devices vary over different time slots due to the mobility of the UAV. 
Hence, device scheduling should be optimized along with the time allocation and UAV trajectory to account for time-varying channel conditions and the target convergence accuracy required of FL.
In the next subsection, we elaborate on the communication and computation models of FL over UAV-enabled networks.

\subsection{UAV-enabled Transmission}

Without loss of generality, we consider a three-dimensional (3D) Cartesian coordinate.
Each ground device is fixed to the ground.
The horizontal coordinate of ground device $k$ is denoted as $\mathbf u_k=[\tilde{x}_k, \tilde{y}_k]\in \mathbb{R}^{1 \times 2}$ with $\tilde{x}_k$ and $\tilde{y}_k$ being $x$- and $y$-coordinates, respectively. 
We assume each ground device is known from the UAV, which flies at a fixed height $H$ m above the ground.
In practice, $H$ corresponds to the minimum altitude at which obstacles can be avoided without needing frequent aircraft ascended and descended.

In this paper, we suppose that the total number of communication rounds denoted by $N$ is fixed.
We focus on the time required for the UAV to complete the FL process, wherein the corresponding variable is denoted as $T$.
To facilitate trajectory design, we adopt the time-discretization technique to divide the mission duration $T$ into $N$ unequal time slots, wherein the $n$-th slot duration is expressed as $\delta[n]$ with $\sum_{n=1}^{N}\delta[n] = T$. 
Note that $\{\delta[n]\}$ and $T$ are variables in this paper.
We denote the location of the UAV at time slot $n$ projected on the horizontal (ground) plane as $\bm{q}[n]=[x[n], y[n]]\in \mathbb{R}^{1\times 2}$, with $x[n]$ and $y[n]$ being $x$- and $y$-coordinates at time slot $n$, respectively.
In addition, we consider the scenario where the initial location of the UAV is pre-deployed, the horizontal coordinate of which is denoted as $\bm{q}_{F} = [x_F, y_F]\in \mathbb{R}^{1\times 2}$. 
This is because the initial location corresponds to the practical mobility constraint on the UAV such as the UAV’s launching location or its pre-mission flying paths, etc.
We assume that the link between each device and the UAV is dominated by the LoS channel because UAVs usually fly at high altitudes for safety reasons, as most existing studies on UAVs \cite{Wu2018MultiUAV, Zhou2018UAVedge, Fu2021UAV}.
Following the free-space path loss model \cite{Wu2018MultiUAV}, the channel gain between device $k$ and the UAV at time slot $n$ is given by 
\begin{eqnarray}
	{h}_k[n]= \beta_0 d_k^{-2}[n] = \frac{\beta_0}{{H^2+\|\bm{q}[n]-\mathbf {u}_k\|_2^2}},
\end{eqnarray}
where $\beta_0$ represents the channel gain at the reference distance of 1 meter and $d_k[n]$ is the distance between device $k$ and the UAV at time slot $n$. 

The distance moved by the UAV in the $n$-th slot is subject to the following constraint
\begin{eqnarray}\label{Cons:Q distance}
	\|\bm q [n]\!-\!\bm q[n-1]\|_2 \!\leq\! \min\{V_{\max}\delta[n], \Delta q_{\max}\},n = 1,\ldots, N, 
\end{eqnarray}
where $V_{\max}$ is the maximum UAV speed, and $\Delta q_{\max}$ is an appropriately chosen value with $\Delta q_{\max}\ll H$ such that channel gain during each round is assumed to be unchanged.

\begin{figure}
	\centering
	\includegraphics[scale = 0.25]{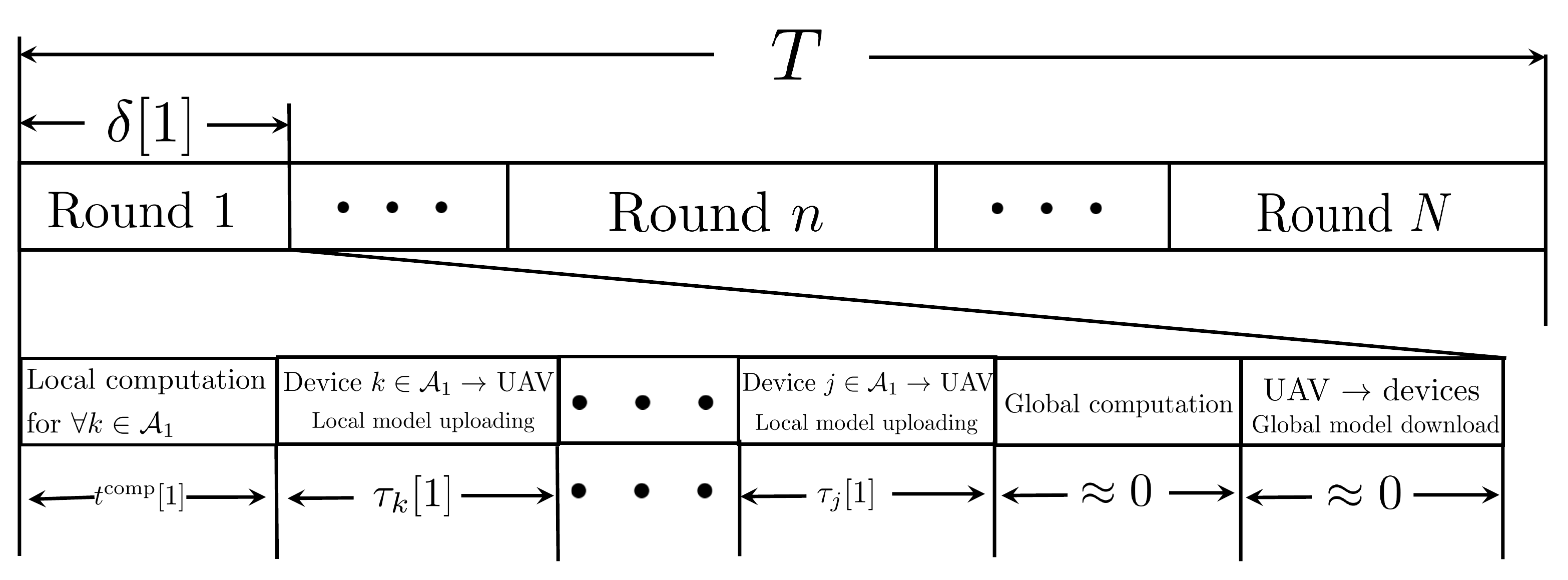}
	\vspace{-6mm}
	\caption{The illustration of UAV-enabled transmission scheme to support FL.}\label{Fig:TDMA model}
\end{figure}

Each time interval $\delta[n]$ consists of four stages, namely, the local computation stage, the local model uploading stage, the global computation stage, and the global model downloading stage, as shown in Fig. \ref{Fig:TDMA model}.
\subsubsection{Local Computation}
We denote the required number of processing cycles for computing one data sample at device $k$ by $c_k$, which can be measured offline and is known as a prior. Since all samples $\{\bm x_{ki},y_{ki}\}_{i \in \mathcal{D}_k}$ have the same size (i.e., number of bits), the number of CPU cycles required for device $k$ to run one local round is $c_k D_k$. 
By denoting the CPU-cycle frequency of the device $k$ by $f_k$, the computation time per round of device $k$ is given by
\begin{eqnarray}\label{T: Tcomp}
	t_k^{\text{comp}} = \frac{c_k D_k}{f_k}.
\end{eqnarray}
The CPU energy consumption of device $k$ for the $n$-th round of local  computation is \cite{Dinh2021FL_TON}
\begin{eqnarray}\label{E:Ecomp}
	E_{k}^{\text{comp}}[n] = a_{k}[n]\sum_{i=1}^{ c_k D_k} \frac{\alpha_k }{2} f_k^2 = a_{k}[n]\frac{\alpha_k }{2} c_k D_k f_k^2, 
\end{eqnarray}
where $\alpha_k/2$ is the effective capacitance coefficient of device $k$'s computing chipset.
Similar to \cite{Dinh2021FL_TON}, we consider the synchronous operation that requires all scheduled devices to simultaneously train their local models and complete their training before entering the communication phase.
Thus, the time cost for local model training at the $n$-th round is given by \cite{Dinh2021FL_TON}
\begin{eqnarray}\label{sumT:sumTcomp}
	t^{\text{comp}}[n] = \max_k \{a_{k}[n]t_k^{\text{comp}}\}.
\end{eqnarray}
Hence, the computing time in one local round is determined by scheduled devices with large date sizes and low CPU frequency.

\subsubsection{Local Model Uploading}
In the local model uploading stage, to avoid interference among devices during the uploading process, TDMA  \cite{Dinh2021FL_TON} is adopted as shown in Fig. \ref{Fig:TDMA model}. 
Specifically, in the local model uploading stage, $|\mathcal{A}_n|$ devices send their respective local models one by one during slot $n$.
The achievable data rate (bit/s) of devices $k$ at the $n$-th round is defined as 
\begin{eqnarray}\label{R:rate1}
	r_{k}[n] = B\log_2\Big(1+\frac{p_k[n]h_k[n]}{\sigma^2}\Big),
\end{eqnarray}
where $B$ is the system bandwidth in Hertz, $p_k[n]$ is the transmit power of device $k$ at the $n$-th round, and $\sigma^2$ is the power of the additional complex Gaussian noise.
Let $\tau_{k}[n]$ denote the duration in which device $k$ transmits its local model to the UAV at the $n$-th round. 
We assume that the dimension of the model parameter vector is fixed throughout the FL process, denoted by $s$ (in bits). 
The uploading time of device $k$ can be calculated as
\begin{eqnarray}\label{R:rate2}
	\tau_{k}[n]	= \frac{a_k[n]s}{r_k[n] }.
\end{eqnarray}
With a given data size, the larger the channel gain, the shorter the required uploading time. 
The UAV can shorten the uploading time of scheduled devices by establishing LoS connections and shortening communication distances with the devices.
By combining \eqref{R:rate1} and \eqref{R:rate2}, to transmit $s$ within a time duration $\tau_k[n]$, the device $k$'s energy consumption is
\begin{eqnarray}\label{E:Ecomm}
	E_k^{\text{comm}}[n] =  \frac{\tau_k[n]\sigma^2}{h_k[n]}\Big(2^{\frac{a_k[n]s}{B\tau_k[n]}}-1\Big).
\end{eqnarray}

After all scheduled devices upload their local training models $\{\bm w_{k,n}\}$ at the $n$-th slot, the UAV performs the global model update and broadcasts the updated result of length $s$ bits to the devices.
Let $f_0$  denote the UAV processor’s fixed computing speed, $c_0$ denote the required number of processing cycles for one local model at the UAV, and $P_0$ denote the transmit power of the UAV. 
The time that spent on updating global model and feeding the result back to devices at the UAV is expressed as
\begin{eqnarray}\label{Eq:UAV time}
	t^{\text{UAV}}[n] = \frac{Kc_0}{f_0} + \frac{s}{ \min_{k \in \mathcal{A}_n}B\log_2\Big(1+\frac{P_0h_k[n]}{\sigma^2}\Big) }.
\end{eqnarray}
One aspect, comparing to the time of local model computation at any device $k$ denoting $\frac{D_kc_k}{f_k}$, the first term in \eqref{Eq:UAV time} is negligible. 
The reasons are that the UAV has a higher computation capability $f_0$ than that of the devices $f_k$ \cite{Zhou2018UAVedge}.
Moreover, the number of devices $K$ is much smaller than the size of the local dataset $D_k$.
The other aspect, the UAV only needs to broadcast the global model to devices while all the scheduled devices upload their local models via TDMA.
Accordingly, we can infer that $t^{\text{UAV}}[n] \ll \sum_{k=1}^{K} \tau_{k}[n] + t^{\text{comp}}[n]$, and the global computation time and the global model downloading time at the UAV are neglected as in \cite{Dinh2021FL_TON}, \cite{Zhou2018UAVedge}.
Thus, the total uploading time of devices can occupy the rest of the time after local model training at the $n$-th round, i.e.,
\begin{eqnarray}\label{comp time}
	\sum_{k=1}^{K} \tau_{k}[n] \leq \delta[n] - t^{\text{comp}}[n], \forall n.
\end{eqnarray}

\section{Convergence Analysis and Problem Formulation} \label{Section:Convergence}
In this section, we conduct the convergence analysis of UAV-enabled FL, which characterizes the impact of key system parameters on the convergence performance, and then formulate a completion time minimization problem, taking into account communication resource constraints, UAV deployment constraints, and target convergence accuracy of FL.

\subsection{Convergence Analysis of FL}\label{subSection:Convergence analysis}

We follow the literature and make two standard assumptions on the loss function and local gradients as follows \cite{bottou2018optimization, Wang2019Adaptive, Allen2018natasha, Zeng2022Wirelessly, Wang2022IRS, Yu2019parallel, Liu2021Privacy}.
\begin{assumption}\label{Assume:Smoothness}
	({\it{Smoothness}}). 
	The function $\mathcal{F}$: $\mathbb{R}^{d}\rightarrow\mathbb{R}$ is $L$-{\it{smooth}}. That is, $\forall \bm w, \bm w'\in \mathbb{R}^{d}$, 
	\begin{eqnarray}
		\!\!\mathcal{F}(\bm w) \leq \mathcal{F}(\bm w') + \langle\nabla \mathcal{F}(\bm w'), \bm w - \bm w'\rangle + \frac{L}{2}\|\bm w'-\bm w\|_2^2.\label{Neq:Smooth}
	\end{eqnarray}
\end{assumption}
And the loss function does not have to satisfy the convexity assumption, and only needs to be lower-bounded, which is the minimal assumption required for convergence \cite{Allen2018natasha, Zeng2022Wirelessly}. 
\begin{assumption}\label{Assume:Bounded loss}
	({\it{Bounded Loss Function}}). 
	For any parameter vector $\bm w$, the loss function $\mathcal{F}(\bm w)$ is lower-bounded by $F^{\star}$.
\end{assumption}

The following assumption is also standard in the literature \cite{Wang2022IRS, Yu2019parallel, Liu2021Privacy}.
\begin{assumption}\label{Assume:Bounded local grad norm}
	({\it{Bounded Sample-wise Gradient Norm}}). For a constant $\kappa>0$,
	the sample-wise gradients at local devices are bounded 
	\begin{eqnarray}\label{Neq:Bounded grad norm}
		\|\nabla f(\bm w_{n};\bm x, y)\|_2^2\leq \kappa.	
	\end{eqnarray}
\end{assumption}
We use the average gradient norm as an indicator of convergence for FL, which is widely adopted in the convergence analysis for non-convex loss function \cite{reisizadeh2020fedpaq,bottou2018optimization}.
According to the above assumptions, the average gradient norm can be bound, explained as follows.
\begin{theorem}\label{theorem:convergence result}
	Suppose that the loss functions satisfy assumptions 1-3, given the learning rate $0\leq\eta \leq \frac{1}{L}$, after $N$ rounds, the average norm of the global gradients is upper bounded by
	\begin{eqnarray}\label{Neq:Convergence bound}
		\frac{1}{N}\sum_{n=0}^{N-1}\|\nabla \mathcal{F}(\bm w_{n})\|^2
		&\leq& \frac{2}{N\eta}(\mathcal{F}(\bm w_0)- F^{\star}) 
		+ \frac{4K\kappa}{ND^2}\sum_{n=0}^{N-1}\sum_{k=1}^K(1-a_k[n+1])D_k^2.
	\end{eqnarray}
\end{theorem}

Given rounds $N$, the FL algorithm achieves an $\epsilon$-approximation solution if \cite{bottou2018optimization}
\begin{eqnarray}\label{Neq:Convergence accuracy}
	\frac{1}{N}\sum_{n=0}^{N-1}\|\nabla \mathcal{F}(\bm w_{n})\|_2^2 \leq \epsilon,
\end{eqnarray}
where $\epsilon>0$ is the convergence threshold. \begin{remark}
	As we can see, the expression \eqref{Neq:Convergence bound} comprises two terms. The first term is the initial optimality gap. For the first term, when the number of communication rounds, the learning step-size, and the initial model parameter are given, the initial optimality gap is a constant.
	The second term is the time-average aggregation error resulting from the effect of device scheduling.
	For the second term, the time-average aggregation error decreases as the number of scheduled devices increases, particularly when large data-size devices are scheduled.
	However, scheduling more devices in one round for model uploading increases the communication burden and may significantly slow down the model aggregation process, especially if there are communication stragglers. 
	Thus, it is crucial to schedule a proper subset of devices to balance training performance and time consumption.
\end{remark}

\subsection{Problem Formulation}
Let $\bm A = \{ a_{k}[n],\forall n\in\mathcal{N},\forall k\in\mathcal{K}\}$, $\bm \Gamma = \{\tau_k[n],\forall n\in\mathcal{N},\forall k\in\mathcal{K}\}$, $\bm Q = \{\bm q[n],\forall n = 0, \ldots, N \}$, and $\bm \delta = \{\delta[n],\forall n = 1, \ldots, N \}$.
In this paper, given the number of communication rounds $N$, we aim to minimize the completion time of the FL training process under the target convergence accuracy requirements by jointly optimizing the devices' scheduling variables $\bm A$, time allocation $\bm \Gamma$,  UAV trajectory $\bm Q$, and slot intervals $\bm \delta$.
To satisfy the convergence requirement in \eqref{Neq:Convergence accuracy}, we ensure that the upper bound in Theorem \ref{theorem:convergence result} is less than $\epsilon$.
Therefore, the corresponding optimization problem is formulated as 
\begin{subequations}\label{Problem: original}
	\begin{eqnarray}
		\!\!\!\!\!\!\!\! \mathop{\text{minimize}}_{\bm A, \bm \Gamma, \bm Q, \bm \delta} &&  \sum_{n=1}^N\delta[n] \nonumber \\
		\!\!\!\!\!\!\!\!	\text{subject to}
		&& a_k[n] \in \{0,1\}, \forall k \in \mathcal{K}, \forall n \in \mathcal {N}, \label{Cons:A binary} \\
		&& \sum_{n=1}^{N}\Big(E_k^{\text{comm}}[n] +   E_{k}^{\text{comp}}[n]\Big) \leq{E}_{k},  \forall k \in \mathcal{K}, \label{Cons:E budget} \\
		&& \sum_{k=1}^{K} \tau_{k}[n]  \!+\!  \max_k a_{k}[n]t_k^{\text{comp}}\leq\delta[n], \forall n \in \mathcal {N}, \label{Cons:T slot} \\
		&&\frac{2}{N\eta}(\mathcal{F}(\bm w_{0})-F^{\star})
		+\frac{4K\kappa}{ND^2}\sum_{n=0}^{N-1}\sum_{k=1}^K(1-a_k[n+1])D_k^2 
		\leq \epsilon, \label{Cons:F accuracy}\\ 
		&&\bm{q}[0] = \bm q_{F},  \label{Cons:Q position}\\
		&& \text{Constraints}~\eqref{Cons:Q distance}.
	\end{eqnarray}
\end{subequations}
Constraints \eqref{Cons:A binary} are integer constraints with respect to the devices' scheduling variables.
Constraints \eqref{Cons:E budget} represent the devices' energy budgets.
Constraints \eqref{Cons:T slot} ensure that the communication and computation time per round cannot exceed the duration of each time slot.
Constraint \eqref{Cons:F accuracy} denotes the target convergence accuracy requirement of FL.
Constraint \eqref{Cons:Q position} is subject to the initial location of the UAV. 
Problem \eqref{Problem: original} is challenging to solve due to the following reasons.
First, the optimization variables $\bm A$ for device scheduling are binary and thus \eqref{Cons:A binary}-\eqref{Cons:F accuracy} involve integer constraints. Second, $E_k^{\text{comm}}[n]$ in constraints \eqref{Cons:E budget} are not jointly convex with respect to the optimization variables $\bm A, \bm \Gamma, \bm Q$, and $\bm \delta$.
Therefore, problem \eqref{Problem: original} is a mixed-integer non-convex problem, which is difficult to be optimally solved in general.

\subsection{Problem Analysis}

In this subsection, we take a close look at problem \eqref{Problem: original} and compare our formulation with those in the existing literature.
Furthermore, we discuss properties of the considered problem formulation and emphasize the necessity of introducing flying UAVs into the FL system.

The existing works on wireless FL with static terrestrial BS \cite{Xia2020Multi, Chen2021FL_TWC, Ma2021FedSA} have made attempts to minimize the completion time under the convergence accuracy requirements by optimizing device scheduling.
Comparing our proposed formulation with the frameworks in \cite {Xia2020Multi, Chen2021FL_TWC, Ma2021FedSA}, the following differences can be seen:
\begin{itemize}
	\item Problem \eqref{Problem: original} not only considers the learning and adaptive device scheduling, but also takes into account the practical device's energy budget and dynamic channel conditions.
	Conversely, \cite{Xia2020Multi, Chen2021FL_TWC}, and \cite{Ma2021FedSA} do not study these practical communication issues.
	
	\item Given the number of communication rounds, problem \eqref{Problem: original} explicitly integrates the target convergence accuracy requirement (i.e., constraint \eqref{Cons:F accuracy}) into the completion time minimization.	
	In contrast, the formulated problems in \cite{Xia2020Multi, Chen2021FL_TWC} do not directly show the relationship between learning convergence accuracy and the completion time.	
\end{itemize}

Prior to solving problem \eqref{Problem: original}, we should check its feasibility. 
According to \eqref{Neq:Convergence bound}, when all devices are scheduled in every communication round, the convergence bound reaches the minimum value.
Therefore, if the minimum value of the convergence bound is larger than the target convergence accuracy, problem \eqref{Problem: original} is infeasible.
Moreover, if $E_k, \forall k$ is not big enough, the constraints in \eqref{Cons:E budget} may not be satisfied even though the scheduled devices consume the lowest amount of energy.
To sum up, the following proposition is given to check the feasibility of problem \eqref{Problem: original}. 
\begin{proposition} \label{proposition:ProblemFeasibility}
	Problem \eqref{Problem: original} is feasible if and only if
	\begin{eqnarray}
		N \geq \lceil\frac{2(\mathcal{F}(\bm w_{0})-F^{\star})}{\epsilon \eta}\rceil\text{~and~} \sum_{k=1}^{K} E_k\geq|\mathcal{S}|  \frac{\sigma^2 sH^2\ln 2}{\beta_0B},
	\end{eqnarray}
	where $|\mathcal{S}|$ denotes the required amount of scheduled devices during  training process, given by
	$	|\mathcal{S}| =\lceil KN-(\epsilon-\frac{2}{N\eta}(\mathcal{F}(\bm w_0)-F^{\star}))\frac{ND^2}{4K\kappa \max_kD_k^2}\rceil.
	$
\end{proposition}
\begin{proof}
	Please refer to Appendix \ref{Proposition:ProblemFeasibility proof}.
\end{proof}

We show the effect of the target convergence accuracy of FL on the completion time.
\begin{theorem} \label{theorem:ProblemOriginal}
	For problem \eqref{Problem: original}, the objective value is non-increasing with $\epsilon$.
\end{theorem}

\begin{proof}
	Please refer to Appendix \ref{theorem:ProblemOriginal proof}.
\end{proof}
Theorem \ref{theorem:ProblemOriginal} sheds light on the tradeoff between the completion time of FL and the target convergence accuracy of FL due to the limited energy budgets.
This is because imposing more stringent convergence accuracy requirements increases the required amount of scheduled devices in the FL training process, which can be observed from constraint \eqref{Cons:F accuracy}.
From \eqref{Problem: original}, we see that time used for the update of the local models depends on the device scheduling matrix $\bm A$ and UAV trajectory $\bm Q$.
Fortunately, by exploiting the UAV's mobility, UAV can dynamically establish favorable connections with scheduled devices in each round to mitigate the communication straggler issue, thereby reducing the completion time of the FL.

\section{Device Scheduling and Time allocation along with Trajectory Design} \label{Section:algorithm}
In this section, we present the BCD-LD method to design the device scheduling, time allocation, and UAV trajectory.
Specifically, we first optimize the device scheduling and time allocation for a given UAV trajectory.
Given device scheduling and time allocation solutions, we then optimize the UAV trajectory $\bm Q$.
To avoid the high computational complexity, we also use LD methods to derive the optimal and closed-form solutions for each optimization subproblem.

\subsection{Joint Device Scheduling and Time Allocation}
To eliminate the max function in constraints \eqref{Cons:T slot}, it is equivalently converted into
\begin{eqnarray}\label{Cons:T slot-A}
	\sum_{k=1}^{K} \tau_{k}[n]  +   a_{k}[n]t_k^{\text{comp}}\leq \delta[n], \forall k, \forall n \in \mathcal {N}.	
\end{eqnarray}
For constraints \eqref{Cons:F accuracy}, it can be rewritten as
\begin{eqnarray} \label{Cons:accuracy-A}
	\sum_{n=0}^{N-1}\sum_{k=1}^Ka_k[n+1]D_k^2 \geq C. 
\end{eqnarray}
where $C = N\sum_{k=1}^KD_k^2\!-\! (\epsilon \!-\! \frac{2}{N\eta}(\mathcal{F}(\bm w_{0})\!-\!F^{\star}) )\frac{ND^2}{4K\kappa}$ is constant.
Constraints \eqref{Cons:Q distance} with given trajectory $\bm Q$ can be reduced to
\begin{eqnarray}\label{Cons:Q speed-A}
	\delta[n] \geq \frac{\|\bm q [n]-\bm q[n-1]\|_2}{V_{\max}}, \forall n.
\end{eqnarray}

By further relaxing the integer constraints, i.e., \eqref{Cons:A binary}, Problem \eqref{Problem: original} under fixed UAV trajectory $\bm Q$ reduces to the following problem,
\begin{subequations}\label{Subproblem:scheduling and time}
	\begin{eqnarray}
		\!\!\!\!\!\!\!\!\!\!\mathop{\text{minimize}}_{\bm A, \bm \delta, \bm \tau} &&  \sum_{n=1}^N\delta[n] \nonumber \\
		\!\!\!\!\!\!\!\!\!\!	\text{subject to}
		&& \sum_{n=1}^{N}\Big(\frac{\tau_k[n] \sigma^2}{h_k[n]}\Big(2^{\frac{a_{k}[n]s}{B\tau_k[n]}}-1\Big) 
		+   a_{k}[n]E_{k}^{\text{comp}}[n]\Big) \leq {E}_{k},  \forall k \in \mathcal{K}, \label{Cons:E budget-A} \\		
		&& 0\leq a_k[n] \leq 1, \forall k \in \mathcal{K}, \forall n \in \mathcal {N}, \label{Cons:binary-relaxing}\\
		&& \text{Constraints~} \eqref{Cons:T slot-A}, \eqref{Cons:accuracy-A}, \text{and~} \eqref{Cons:Q speed-A}. \nonumber
	\end{eqnarray}
\end{subequations}
Define $\frac{\tau_k[n] \sigma^2}{h_k[n]}\Big(2^{\frac{a_{k}[n]s}{B\tau_k[n]}}-1\Big) = 0$ when $\tau_k[n] = 0, \forall k, \forall n$, such that the left-hand-sides of constraints \eqref{Cons:E budget-A} are continuous with respect to $\tau_k[n]$ with $\tau_k[n]\geq 0$.
Accordingly, since $\frac{ \sigma^2}{h_k[n]}\big(2^{\frac{a_{k}[n]s}{B}}-1\big)$
is convex with respect to $a_{k}[n]$, its perspective function 
$\frac{\tau_k[n] \sigma^2}{h_k[n]}\Big(2^{\frac{a_{k}[n]s}{B\tau_k[n]}}-1\Big)$ with $\tau_k[n]\geq0$ in constraints \eqref{Cons:E budget-A} is jointly convex with respect to $a_{k}[n]$ and $\tau_k[n]$.
Furthermore, all remaining constraints in problem \eqref{Subproblem:scheduling and time} are affine constraints.
As such, it is concluded that problem \eqref{Subproblem:scheduling and time} is convex, which can be optimally solved by the Lagrange duality method.
The partial Lagrange function of  problem \eqref{Subproblem:scheduling and time} is given by
\begin{eqnarray*}\label{Eq:Lagrange function}
	\mathcal{L}	(\bm \delta, \bm A, \bm \tau, \bm \lambda, \bm \mu, \xi) 
	&=& \sum_{n=1}^N\delta[n] 
	+ \sum_{k=1}^K \lambda_k \Big(
	\sum_{n=1}^{N}\Big(\frac{\tau_k[n] \sigma^2}{h_k[n]}\Big(2^{\frac{a_{k}[n]s}{B\tau_k[n]}}-1\Big)+   a_{k}[n]E_{k}^{\text{comp}}[n]\Big)  - {E}_{k}\!\Big)\nonumber \\
&\!\!\!\!\!\!\!\!\!\!\!\!\!\!\!\!\!\!+& \!\!\!\!\!\!\!\!\!\!\sum_{n=1}^N \sum_{k=1}^{K}\mu_{k}[n]\Big( \sum_{k=1}^{K} \tau_{k}[n] +   a_{k}[n]t_k^{\text{comp}}- \delta[n] \Big) 
	+ \xi (C-\sum_{n=0}^{N-1}\sum_{k=1}^Ka_k[n+1]D_k^2),
\end{eqnarray*}
where $\bm \lambda = \{\lambda_k, \forall k\}$, $\bm \mu = \{\mu_k[n], \forall k, \forall n\}$, and $\xi$ are the non-negative Lagrange multipliers associated with constraints \eqref{Cons:E budget-A}, \eqref{Cons:T slot-A}, and \eqref{Cons:accuracy-A}, respectively. 
The other boundary constraints in problem \eqref{Subproblem:scheduling and time} will be absorbed into the optimal solution in the following. 
Accordingly, the dual function is given by
\begin{eqnarray}\label{Subproblem:dual function}
	g(\bm \lambda, \bm \mu, \xi)	 = \mathop{\text{minimize}}_{\bm \delta, \bm A, \bm \tau, \xi}&&	\mathcal{L}(\bm \delta, \bm A, \bm \tau, \bm \lambda, \bm \mu) \nonumber \\
	\text{subject to} 
	&& \text{Constraints~} \eqref{Cons:binary-relaxing} ~\text{and}~\eqref{Cons:T slot-A}. 
\end{eqnarray}
To make $g(\bm\lambda, \bm\mu, \xi)$ lower-bounded (i.e., $g(\bm\lambda, \bm\mu, \xi)>-\infty$), the condition $\sum_{k=1}^{K}\mu_{k}[n]\leq 1$ must hold.
Therefore, the dual problem of problem 
\eqref{Subproblem:scheduling and time} is given by
\begin{subequations}\label{Subproblem:dual problem-A}
	\begin{eqnarray}
		\mathop{\text{maximize}}_{\bm \lambda, \bm \mu, \xi}&&	g(\bm \lambda, \bm \mu, \xi) \nonumber \\
		\text{subject to} 
		&&\sum_{k=1}^{K}\mu_{k}[n]>1,\forall n, \label{Cons: mu}\\
		&& \bm \lambda \succeq 0, \bm \mu \succeq 0, \xi \geq 0. \label{Cons: mu-lambda}
	\end{eqnarray}
\end{subequations}
Since the strong duality holds, we can solve problem \eqref{Subproblem:scheduling and time} by equivalently solving its dual problem \eqref{Subproblem:dual problem-A}.
In the following, we show how to obtain the dual function $g(\bm \lambda, \bm \mu, \xi)$ by solving problem \eqref{Subproblem:dual function} under any given $\bm \lambda, \bm \mu, \xi$ that satisfy \eqref{Cons: mu} and \eqref{Cons: mu-lambda}.
By ignoring the constant terms, problem \eqref{Subproblem:scheduling and time} reduces to the following formulation,
\begin{subequations}\label{Subproblem:dual function-1}
	\begin{eqnarray}
		\mathop{\text{minimize}}_{\bm \delta, \bm A, \bm \tau}&&\sum_{n=1}^N\Big(1-\sum_{k=1}^{K}\mu_{k}[n]\Big)\delta[n]
		+ \!\sum_{k=1}^K\! \lambda_k 
		\!\!\sum_{n=1}^{N}\!\!\Big(\!\frac{\tau_k[n] \sigma^2}{h_k[n]}\!\Big(\!2^{\frac{a_{k}[n]s}{B\tau_k[n]}}\!-\!1\Big)\!+ \!  a_{k}[n]E_{k}^{\text{comp}}[n]\!\Big)  \nonumber \\
		&&   + \sum_{n=1}^N \sum_{k=1}^{K}\mu_{k}[n]\Big(\sum_{k=1}^{K} \tau_{k}[n] +   a_{k}[n]t_k^{\text{comp}}\Big) 
		 - \xi \sum_{n=1}^{N}\sum_{k=1}^Ka_k[n]D_k^2 \nonumber \\
		\text{subject to} 
		&& \text{Constraints~} \eqref{Cons:binary-relaxing} ~\text{and}~\eqref{Cons:T slot-A},
	\end{eqnarray}
\end{subequations}

Note that problem \eqref{Subproblem:dual function-1} can be decomposed into two sets of subproblems, i.e., optimizing $\bm \delta$ and jointly optimizing $\bm A$ and $\bm \tau$. 
We first consider the following optimization problem for optimizing $\bm A$ and $\bm \tau$, 
	\begin{eqnarray}\label{Subproblem:Primal A and tau}
		\mathop{\text{minimize}}_{\bm A, \bm \tau} &&\sum_{n=1}^N \sum_{k=1}^{K}\mu_{k}[n]\Big(\sum_{k=1}^{K} \tau_{k}[n] +  a_{k}[n]t_k^{\text{comp}}\Big)	 \nonumber \\
		&& + \sum_{k=1}^K \lambda_k 
		\sum_{n=1}^{N}\Big(\frac{\tau_k[n] \sigma^2}{h_k[n]}\Big(2^{\frac{a_{k}[n]s}{B\tau_k[n]}}-1\Big)+  a_{k}[n]E_{k}^{\text{comp}}[n]\Big) 
	 - \xi \sum_{n=1}^{N}\sum_{k=1}^Ka_k[n]D_k^2 \nonumber \\
		\text{subject to} 
		&& \text{Constraints~} \eqref{Cons:binary-relaxing},
	\end{eqnarray}
Problem \eqref{Subproblem:Primal A and tau} is convex with respect to $\{a_k[n]\}$ and $\{\tau_k[n]\}$, so the optimal solution is the one that satisfies the Karush-Kuhn-Tucker (KKT) conditions. 
To obtain the optimal $\bm \tau$ in problem \eqref{Subproblem:Primal A and tau}, we have the following theorem.
\begin{theorem} \label{theorem:tau-solution}
By setting the first derivative of the objective function of problem \eqref{Subproblem:Primal A and tau} with respect to $\tau_k[n]$ to zeros, the optimal time allocation can be written as 
	\begin{eqnarray} \label{Eq:Solution tau}
		\tau_{k}^{\star}[n] = a_k[n]\Bigg[ \frac{s\ln2}{B(1+\mathcal{W}(\frac{h_k[n]\sum_{k=1}^{K}\mu_{k}[n]}{\lambda_k \sigma^2e}-\frac{1}{e})}\Bigg]^+\!\!\!, \forall k, \forall n,
	\end{eqnarray}	
	where $\mathcal{W}$ is the Lambert W function with $\mathcal{W}(x)e^{\mathcal{W}(x)} = x$ and $[x]^+ \triangleq \max\{x,0\}$. 
\end{theorem}
\begin{proof}
	Please see Appendix \ref{theorem:tau-solution proof}.
\end{proof}
By substituting the obtained $\{\tilde{\tau}_{k}[n]\}$ into problem \eqref{Subproblem:Primal A and tau}, we have
\begin{subequations}\label{Subproblem:Primal A}
	\begin{eqnarray}
		\mathop{\text{minimize}}_{\bm A} &&	\sum_{k=1}^K
		\sum_{n=1}^{N}g(\tilde{\tau}_{k}[n])a_k[n] \nonumber \\
		\text{subject to} 
		&& 0\leq a_k[n] \leq 1, \forall k \in \mathcal{K}, \forall n \in \mathcal {N}, 
	\end{eqnarray}
\end{subequations}
where
$
	\tilde{\tau}_{k}[n] \!=\! \frac{\tau_{k}^{\star}[n]}{a_k[n]}, 
	g(\tilde{\tau}_{k}[n]) \!=\!  \lambda_k \Big(\frac{\tilde{\tau}_{k}[n] \sigma^2}{h_k[n]}\Big(2^{\frac{s}{B\tilde{\tau}_{k}[n]}}-1\Big)+   E_{k}^{\text{comp}}[n]\Big)  
	+ \tilde{\tau}_{k}[n]\sum_{k=1}^{K}\mu_{k}[n] +  \mu_{k}[n]t_k^{\text{comp}}- \xi D_k^2.
$
Problem \eqref{Subproblem:Primal A} is a convex linear programming.
It is evident that, to minimize the objective function in \eqref{Subproblem:Primal A}, the $a_k[n]$ corresponding to $g(\tilde{\tau}_{k}[n])\leq 0$ should be set to 1.
Therefore, the optimal solution for $\{a_k[n]\}$ are thus given by the following theorem.
\begin{theorem} \label{theorem:A-solution}
	For problem \eqref{Subproblem:Primal A}, 
	the optimal device scheduling variables $\bm A$ can be expressed as
	\begin{eqnarray} \label{Eq:A-solution}
		&&a^{\star}_k[n]= 
		\left\{
		\begin{aligned}
			&1,  \text{if~} g(\tilde{\tau}_{k}[n])\leq 0, \\
			& 0, \  \text{otherwise}.	
		\end{aligned}
		\right.
	\end{eqnarray}
\end{theorem}
\begin{remark}
Theorem \ref{theorem:A-solution} shows that the dual method obtains the binary solution for device scheduling variables, which guarantees both optimality and feasibility of the original problem.
It would also make sense to select a device with a large data-set, a large channel gain, and a high computation capacity.
\end{remark}

Next, we consider the other subproblem for optimizing $\bm \delta$, which is given by
\begin{eqnarray}\label{Subproblem:Primal delta}
	\mathop{\text{minimize}}_{\bm \delta} && \sum_{n=1}^N\Big(1-\sum_{k=1}^{K}\mu_{k}[n]\Big)\delta[n] \nonumber\\
	\text{subject to}
	&& \delta[n] \geq \frac{	\|\bm q [n]-\bm q[n-1]\|_2}{V_{\max}}, \forall n.
\end{eqnarray}
It is evident that problem \eqref{Subproblem:Primal delta} is a linear programming. For problem \eqref{Subproblem:Primal delta}, since $\sum_{k=1}^{K}\mu_{k}[n]\leq1$ holds,   the optimal time slot variables are given by
\begin{eqnarray} \label{Eq:delta-solution}
	&&\delta^{\star}[n]= 
	\left\{
	\begin{aligned}
		&\frac{\|\bm q [n]-\bm q[n-1]\|_2 }{V_{\max}}, \text{if } \sum_{k=1}^{K}\mu_{k}[n]<1, \\
		& b, \  \text{if } \sum_{k=1}^{K}\mu_{k}[n]=1,
	\end{aligned}
	\right.
\end{eqnarray}	
where $b$ can be any arbitrary real number which is not smaller than $\frac{\|\bm q [n]-\bm q[n-1]\|_2 }{V_{\max}}$ since the objective function of problem \eqref{Subproblem:Primal delta} is not affected in this case. For simplicity, we set $b = \frac{\|\bm q [n]-\bm q[n-1]\|_2 }{V_{\max}}$. In general, \eqref{Eq:delta-solution} cannot provide the optimal primal solution for problem \eqref{Subproblem:scheduling and time} even with optimal dual variables.
Nevertheless, with the above proposed solutions to problems \eqref{Subproblem:Primal A and tau} and \eqref{Subproblem:Primal delta},
the dual function $g(\bm \lambda, \bm \mu, \xi)$ is obtained.

 The dual problem \eqref{Subproblem:dual problem-A} can be solved after obtaining $(\bm \tau^{\star}, \bm \delta^{\star}, \bm A^{\star})$.
Despite the fact that the dual function $g(\bm \lambda, \bm \mu, \xi)$ is always convex by definition, it is generally non-differentiable. 
To address this challenge, a common subgradient-based method, such as projected subgradient decent method, can be applied to solving the \eqref{Subproblem:dual problem-A}.
Specifically, dual variables are updated in each iteration via
\begin{eqnarray} 
	&&\mu_k^{\frac{1}{2}}[n] = [\mu_k[n] +  \phi\sum_{k=1}^{K} \tau_{k}[n]  +   a_{k}[n]t_k^{\text{comp}}- \delta[n]]^{+}, \mu_k[n] =\frac{	\mu_k^{\frac{1}{2}}[n] }{\max\{1,\sum_{k=1}^{K}	\mu_k^{\frac{1}{2}}[n]\}},\label{Eq:dual mu update}\\
	&&	\lambda_k = [\lambda_k + \phi\sum_{n=1}^{N}\Big(\frac{\tau_k[n] \sigma^2}{h_k[n]}\Big(2^{\frac{a_{k}[n]s}{B\tau_k[n]}}-1\Big)+  a_{k}[n]E_{k}^{\text{comp}}[n]\Big) -{E}_{k}]^{+}, \label{Eq:dual lambda update}  \\ 
	&& \xi = [\xi +  \phi(C-\sum_{n=0}^{N-1}\sum_{k=1}^Ka_k[n+1]D_k^2)]^{+}, \label{Eq:dual xi update}
\end{eqnarray}
where $\phi$ is a dynamically chosen step-size. 
Based on $\bm \lambda^{\star}$, $\bm \mu^{\star}$, and $\xi^{\star}$ have been obtained,  it is left to find the optimal primal solution $\{\bm \delta^{\star}, \bm \tau^{\star}, \bm A^{\star}\}$ of problem  \eqref{Subproblem:scheduling and time}. 
For a convex optimization problem, a solution that minimizes the Lagrange function is optimal only if it is unique and primal feasible. 
Unfortunately, in this case, the solutions $\bm \delta$ that minimizes $\mathcal{L}(\bm \delta, \bm A, \bm \tau, \bm \lambda, \bm \mu, \xi)$ are not unique when $\sum_{k=1}^{K}\mu_{k}[n]\leq1$. 
To construct the optimal solution, additional steps are needed. 
Specifically,  given $\bm \lambda^{\star}$, $\bm \mu^{\star}$, and $\xi^{\star}$, the optimal time allocation (i.e.,  $\bm \tau^{\star}$) and device scheduling variable (i.e.,  $\bm A^{\star}$) can be uniquely obtained by \eqref{Eq:Solution tau} and \eqref{Eq:A-solution}, respectively. 
By substituting $ \bm \tau^{\star}$ and $\bm A^{\star}$ into the primal problem \eqref{Subproblem:scheduling and time}, we have 
\begin{subequations}\label{Subproblem:time}
	\begin{eqnarray}
		\mathop{\text{minimize}}_{ \bm \delta} &&  \sum_{n=1}^N\delta[n] \nonumber \\
		\text{subject to}
		&& \sum_{k=1}^{K} \tau_{k}[n]  +   a_{k}[n]t_k^{\text{comp}}\leq \delta[n], \forall k, \forall n \in \mathcal {N}, \\
		&& \delta[n] \geq \frac{	\|\bm q [n]-\bm q[n-1]\|_2}{V_{\max}}, \forall n.
	\end{eqnarray}
\end{subequations}
It is evident that  problem \eqref{Subproblem:time} is a convex linear programming and can be separated into $N$ independent problems. 
The optimal time slot variables are given by
\begin{eqnarray} \label{Eq:delta-solution1}
		\delta^{\star}[n]= \max\Big\{
	\frac{\|\bm q [n]-\bm q[n-1]\|_2 }{V_{\max}}, \sum_{k=1}^{K} \tau_{k}[n]  +  \max_k a_{k}[n]t_k^{\text{comp}}\Big\}. 
\end{eqnarray}	
Algorithm \ref{Algo:scheduling and time} summarizes the details of obtaining the optimal solution to problem \eqref{Subproblem:scheduling and time}.
Algorithm \ref{Algo:scheduling and time} includes three parts, i.e., solving problems \eqref{Subproblem:Primal A and tau} and \eqref{Subproblem:Primal delta}, updating the dual variables, and solving problem \eqref{Subproblem:time}.
According to \eqref{Eq:Solution tau} and \eqref{Eq:A-solution}, the complexity of solving problem \eqref{Subproblem:Primal A and tau} is $\mathcal{O}(KN)$.
According to \eqref{Eq:delta-solution} and \eqref{Eq:delta-solution1}, the complexity of both solving problem \eqref{Subproblem:Primal delta} and problem \eqref{Subproblem:time} is $\mathcal{O}(N)$.
In addition, according to \eqref{Eq:dual lambda update}-\eqref{Eq:dual mu update}, the complexity of updating the dual variables is $\mathcal{O}(KN)$.
Denote by $l_1$ the number of iterations required for convergence.
Therefore, the total complexity of Algorithm \ref{Algo:scheduling and time} is given by $\mathcal{O}(l_1KN)$.

\begin{algorithm}[t]
	\caption{Dual Method for Problem \eqref{Subproblem:scheduling and time}}\label{Algo:scheduling and time}
	\begin{algorithmic}[1]
		\STATE {\textbf{Input:} $K$, $N$, $\bm Q$}.
		\STATE {Initialize dual variables $\{\lambda_{k} = 1\}$, $ \{\mu_{k}[n] =  1/K\}$.}
		\REPEAT
		\STATE {Update the primal variables $\bm A$, $\bm \tau$, and $\bm \delta$ according to  \eqref{Eq:A-solution}, \eqref{Eq:Solution tau}, and \eqref{Eq:delta-solution}}.
		\STATE {Update the dual variables $\bm \lambda$,  $\bm \mu$, $\xi$ according to \eqref{Eq:dual mu update}-\eqref{Eq:dual xi update}}.
		\UNTIL{$\bm \lambda$ and $\bm \mu$ converge within a prescribed accuracy}	
		\STATE {Update $\bm \delta^{\star}$ according to \eqref{Eq:delta-solution1}}.
		\STATE {\textbf{Output:} $\bm A^{\star}$, ${\bm \tau}^{\star}$, and $\bm \delta^{\star}$.}	 		
	\end{algorithmic}
\end{algorithm}

\subsection{Trajectory Design}
Given device scheduling and time allocation $\{\bm A, \bm \delta, \bm \tau\}$, the trajectory optimization subproblem is reduced to the following feasibility checking problem
\begin{subequations} \label{Problem:Q-feasibility}
	\begin{eqnarray}
		\mathop{\text{find}}_{\bm Q} 
		&&  \bm Q\nonumber \\
		\text{subject to}
		&& \sum_{n=1}^{N}b_k[n] \|\bm q[n]-{{\bf w}_k}\|_2^2 \leq \bar{E}_{k}, \forall k \in \mathcal{K}, \label{Cons:E budget-Q} \\
		&& \text{Constraints}~\eqref{Cons:Q distance}, \eqref{Cons:Q position},
	\end{eqnarray}
\end{subequations}
where $\bar{E}_k \triangleq {E}_{k} - \sum_{n=1}^{N}a_{k}[n]E_{k}^{\text{comp}}[n]-\sum_{n=1}^{N}\frac{\tau_k[n]\sigma^2H^2}{\beta_0}\big(2^{\frac{a_k[n]s}{B\tau_k[n]}}-1\big)$ and $b_k[n]\triangleq \frac{\tau_k[n]\sigma^2}{\beta_0}\big(2^{\frac{a_k[n]s}{B\tau_k[n]}}-1\big)$.
According to \eqref{R:rate2}, $\tau_{k}^{\star}[n]$ is a non-increasing function with respect to channel gain $h_k[n]$.
Hence, to further reduce the time consumption while ensuring the feasibility of problem \eqref{Problem:Q-feasibility}, we transform problem \eqref{Problem:Q-feasibility} into the following problem with an explicit objective function
\begin{subequations} \label{Subproblem:Q-objective}
	\begin{eqnarray}
		\mathop{\text{minimize}}_{\bm Q} 
		&&  \sum_{n=1}^N\sum_{k=1}^Kb_k[n]\|\bm q[n]-{{\bf w}_k}\|_2^2\nonumber \\
		\text{subject to}
		&& \text{Constraints}~\eqref{Cons:Q distance}, \eqref{Cons:Q position}, \eqref{Cons:E budget-Q}.
	\end{eqnarray}
\end{subequations}
Although the convex QCQP problem \eqref{Subproblem:Q-objective} can be solved using a general-purpose solver through interior-point methods,  to further reduce the computational complexity, we exploit the specific structure of problem \eqref{Subproblem:Q-objective} and find its optimal solution by using a Lagrange dual ascent method in the sequel.
The partial Lagrange function of problem \eqref{Subproblem:Q-objective} can be expressed as
\begin{eqnarray}\label{Eq:Lagrange function-Q}
	\mathcal{L}	(\bm Q, \bm \gamma)
	= \sum_{n=1}^N\sum_{k=1}^K(a_k[n]+\gamma_kb_k[n])\|\bm q[n]-{{\bf w}_k}\|_2^2
	-\sum_{k=1}^{K}\gamma_k\bar{E}_{k}, 
\end{eqnarray}
where $\bm \gamma = \{\gamma_k, \forall k\}$ are the non-negative Lagrange multipliers associated with constraints \eqref{Cons:E budget-Q}. 
Accordingly, the dual function is given by
	\begin{eqnarray}\label{Subproblem:dual function-Q}
		g(\bm \gamma)	 = \mathop{\text{minimize}}_{\bm Q}&&	\mathcal{L}	(\bm Q, \bm \gamma) \nonumber \\
		\text{subject to} 
		&& \text{Constraints}~\eqref{Cons:Q distance}, \eqref{Cons:Q position}.
	\end{eqnarray}
The dual problem of problem \eqref{Subproblem:Q-objective} is given by
	\begin{eqnarray}\label{Subproblem:dual problem-Q}
		\mathop{\text{maximize}}_{\bm \gamma\succeq 0}&&	g(\bm \gamma) 
	\end{eqnarray}

As the strong duality holds, we can solve problem \eqref{Subproblem:Q-objective} by equivalently solving its dual problem \eqref{Subproblem:dual problem-Q}. In the following, we show how to obtain the dual function $g(\bm \gamma)$ by solving problem \eqref{Subproblem:dual function-Q} under any given $\bm \gamma \succeq 0$. 
By ignoring the constant terms, problem \eqref{Subproblem:dual function-Q} reduces to 
	\begin{eqnarray}\label{Subproblem:dual function-Q-reduced}
		\mathop{\text{minimize}}_{\bm Q}&&	\sum_{n=1}^N\sum_{k=1}^K(b_k[n]+\gamma_kb_k[n])\|\bm q[n]-{{\bf w}_k}\|_2^2\nonumber \\
		\text{subject to} 
		&& \text{Constraints}~\eqref{Cons:Q distance}, \eqref{Cons:Q position}.
	\end{eqnarray}
Since problem \eqref{Subproblem:dual function-Q-reduced} is convex with respect to $\bm Q$, the solution that satisfies KKT conditions is also optimal. 
To obtain the optimal solution $\bm Q$ in problem \eqref{Subproblem:dual function-Q-reduced}, we have the following theorem.
\begin{theorem} \label{theorem:Q-solution}
	For problem \eqref{Subproblem:dual function-Q-reduced}, the optimal UAV trajectory can be denoted as 
	\begin{eqnarray} \label{Eq:Solution Q}
	\bm q[n]&=& 
		\left\{
		\begin{aligned}
			& {\bf{q}_I}, \text{if~} n=0, \\
			&	\mathcal{P}_{\mathcal{C}} \Big(\frac{\sum_{k=1}^K(b_k[n]+\gamma_kb_k[n])({\bf {w}_k}-\bm q[n-1])}{\sum_{k=1}^K(b_k[n]+\gamma_kb_k[n])}\Big) + \bm{q}[n-1], 
			& \text{otherwise},
		\end{aligned}
		\right.
	\end{eqnarray}	
	where $\mathcal{P}_{\mathcal{C}}(\bm x) \!:=\! \min\big\{{\min\{V_{\max}\delta[n], \Delta q_{\max}\}}\big/{\|\bm x\|},1 \big\}\bm x $ is the projector associated with space $\mathcal{C}$.
\end{theorem} 

After obtaining $\bm Q^{\star}$ for given $\bm \gamma$, we next solve the dual problem \eqref{Subproblem:dual problem-Q} to find the optimal dual variables which maximize $g(\bm \gamma)$. 
Likewise in joint device scheduling and time allocation, the dual variable is determined by the subgradient based method, the updates of which in each iteration can be given by
$	\gamma_k = \Big[\gamma_k + \varphi\Big(\sum_{n=1}^{N}b_k[n] \|\bm q[n]-{\bf w_k}\|_2^2 - \bar{E}_{k}\Big)\Big]^+$,
where $\varphi$ is a dynamically chosen step-size.
The procedure for solving problem \eqref{Subproblem:Q-objective} is similar to that of solving problem \eqref{Subproblem:scheduling and time}. The details are omitted due to space limitations.
\begin{algorithm}[t]
	\caption{BCD-LD for Problem \eqref{Problem: original}}\label{Algo:main}
	\begin{algorithmic}[1]
		\STATE {\textbf{Input:} $M$, $N$.}
		\STATE {Initialize $\bm Q^{0}$. Let $r=1$.}
		\REPEAT
		\STATE {Solve problem \eqref{Subproblem:scheduling and time} by applying Algorithm \ref{Algo:scheduling and time}  for given $\bm Q^{r-1}$}, and denote the optimal solution as $\{\bm A^{r}, \bm \tau^{r}, \bm \delta^{r}\}$.
		\STATE {Solve problem \eqref{Subproblem:Q-objective} for given $\{\bm A^{r}, \bm \tau^{r}, \bm \delta^{r}\}$, and denote the optimal solution as $\bm Q^{r}$}.
		\STATE {Update $r = r+1$}.
		\UNTIL{the fractional decrease of the objective value is below a threshod $\varepsilon$.}	
		\STATE {\textbf{Output:} $\bm A$, $\bm \tau$, $\bm \delta$, and $ \bm Q$.}	 		
	\end{algorithmic}
\end{algorithm}
\subsection{Overall Algorithm Design}
In this section, based on the results in previous two subsections, we present an efficient iterative algorithm to solve problem \eqref{Problem: original} by using the BCD technique \cite{xu2013block}, which is summarized in Algorithm \ref{Algo:main}.
Note that the subproblem for updating each block of variables is solved with optimality in each iteration.
The convergence of Algorithm $\ref{Algo:main}$ is shown in the following proposition.
\begin{proposition}\label{convergence proposition}
	With Algorithm 3, the objective value of problem \eqref{Problem: original} decreases as the number of iterations increases until convergence.
	\vspace{-2mm}	
\end{proposition}
\begin{proof}
	Please refer to Appendix \ref{convergence proposition proof}.
\end{proof}

\section{Numerical Results}\label{Section:simulation}
In this section, we present extensive numerical results to demonstrate the effectiveness of the proposed algorithm for FL via UAV.
We consider the image classification task on the widely-used CIFAR-10 datasets \cite{krizhevsky2009learning}, which is composed of 60, 000 $32 \times32$ RGB color images in 10 different classes.
We adopt the multi-nominal logistic regression model to classify the target datasets with features $d = 32\times 32\times 3\times 10$ and leaning rate $\eta = 0.01$.
And the loss function is 
\begin{eqnarray}
	\mathcal{F}_k(\bm{w})= \frac{-1}{D_k}\sum_{i = 1}^{D_k}\sum_{c= 1}^{10}1_{\{y_i = c\}}\log\frac{\text{exp}(\bm x_i^{\sf T}\bm w_c)}{\sum_{j =1}^{10}\text{exp}(\bm x_i^{\sf T}\bm w_j)}.
\end{eqnarray}
Each parameter is assumed to be stored with $I =32$ bits, i.e., $s = I \times d$.
The effective capacitance coefficient of all devices' computing chipset is $\alpha_1=\ldots =\alpha_k = 10^{-28}$.
The local training data are uniformly drawn from 50,000 images with $\{D_k\}$, and the test dataset contains the remaining 10,000 images. 
Considering different numbers of training samples, the values of $\{D_k\}$ are drawn uniformly.
In the simulations, the service area of the UAV is restricted to a square area with the size of [0, 400] m $\times$ [0, 400] m. 
The UAV is assumed to fly at a fixed altitude of $H = 100$ m, which complies with the practical rule, i.e., commercial UAVs should not fly over 400 feet (122 m) \cite{FAA2016UAV}. 
And the UAV is placed at a predetermined initial position, i.e., (200, 0, 100) m.
The maximum speed of the UAV is ${V}_{\max}=20$ m/s.
In addition, considering channel heterogeneity, $K$ devices are separated into two clusters allocated in two circles, i.e., cluster $A$ with $\lfloor0.3K\rfloor$ devices and cluster $B$ with $K-\lfloor0.3K\rfloor$ devices. 
Specifically, the devices in cluster $A$ and cluster $B$ are randomly and uniformly distributed in circles centered at (100, 100) m and (300, 300) m with a radius of 100 m, respectively.
Other parameters are summarized in Table \ref{table:setupNumerical}.
\begin{table*}[t]
	\centering
	\caption{Parameter settings for simulations.}\label{table:setupNumerical}
	{\renewcommand{\arraystretch}{1.2}\begin{tabular}{|p{3cm}|p{1.8cm}|p{2.3cm}|p{1.9cm}|p{2.5cm}|p{2.0cm}|}
			\hline
			Number of devices & $K = 40$  &	 Noise power  & $-174$ dBm/Hz&CPU frequency & $f_k = 5$ GHz \\
			\hline
			Channel gain at 1 meter & $\beta_0=-50$ dB  	& 	Uplink bandwidth & 10MHz  & Maximum distance  &  $\Delta q_{\max} = 10$ m	 \\
			\hline
			Communication rounds &  $N = 4000$  & CPU cycles & $c_k =10$ & Accuracy of Alg. \ref{Algo:main} &$\varepsilon= 10^{-3}$
			\\
			\hline
	\end{tabular}}
\end{table*}

We compare the performance of the proposed joint design scheme with the following schemes: 
	\textbf{1) Static UAV}: 
With this scheme, the UAV is placed at a predetermined initial position, i.e., (200, 0, 100) m, and remains static during the whole mission time.  
This scheme optimizes $\{\bm A, \bm \tau, \bm \delta\}$ to minimize the completion time for FL by solving problem \eqref{Subproblem:scheduling and time}, as presented in Algorithm \ref{Algo:scheduling and time}.	
	\textbf{2) Static UAV w/ HS}: 
	In this scheme, the UAV is static and placed at a predetermined initial position.  
	With the static UAV, the channel conditions of devices are known and invariant.
	According to the channel-aware criterion \cite{Ren2020Scheduling}, the devices with bigger channel gains are always scheduled to satisfy constraint \eqref{Cons:F accuracy} with equality.
	Therefore, in this scheme, $\bm A$ and $\bm Q$ are deterministic, and we only need to optimize $\{\bm \tau, \bm \delta\}$ to minimize the completion time for FL.	
	\textbf{3) Full Scheduling}: This scheme schedules all devices to upload in every communication round. And without model aggregation loss resulting from device scheduling, this scheme achieves the minimum convergence accuracy of FL, which serves as a performance bound for other schemes.
	By substituting the given $\bm A$ into problem \eqref{Problem: original}, this scheme optimizes $\{ \bm \tau, \bm \delta,\bm Q\}$ to minimize the completion time for FL. We still adopt the BCD-based iterative algorithm to solve the resulting problem, which is similar to Algorithm \ref{Algo:main}.



\begin{figure}[t]
	\centering
	\begin{minipage}{.3\textwidth}
		\centering
		\includegraphics[width=1.1\textwidth]{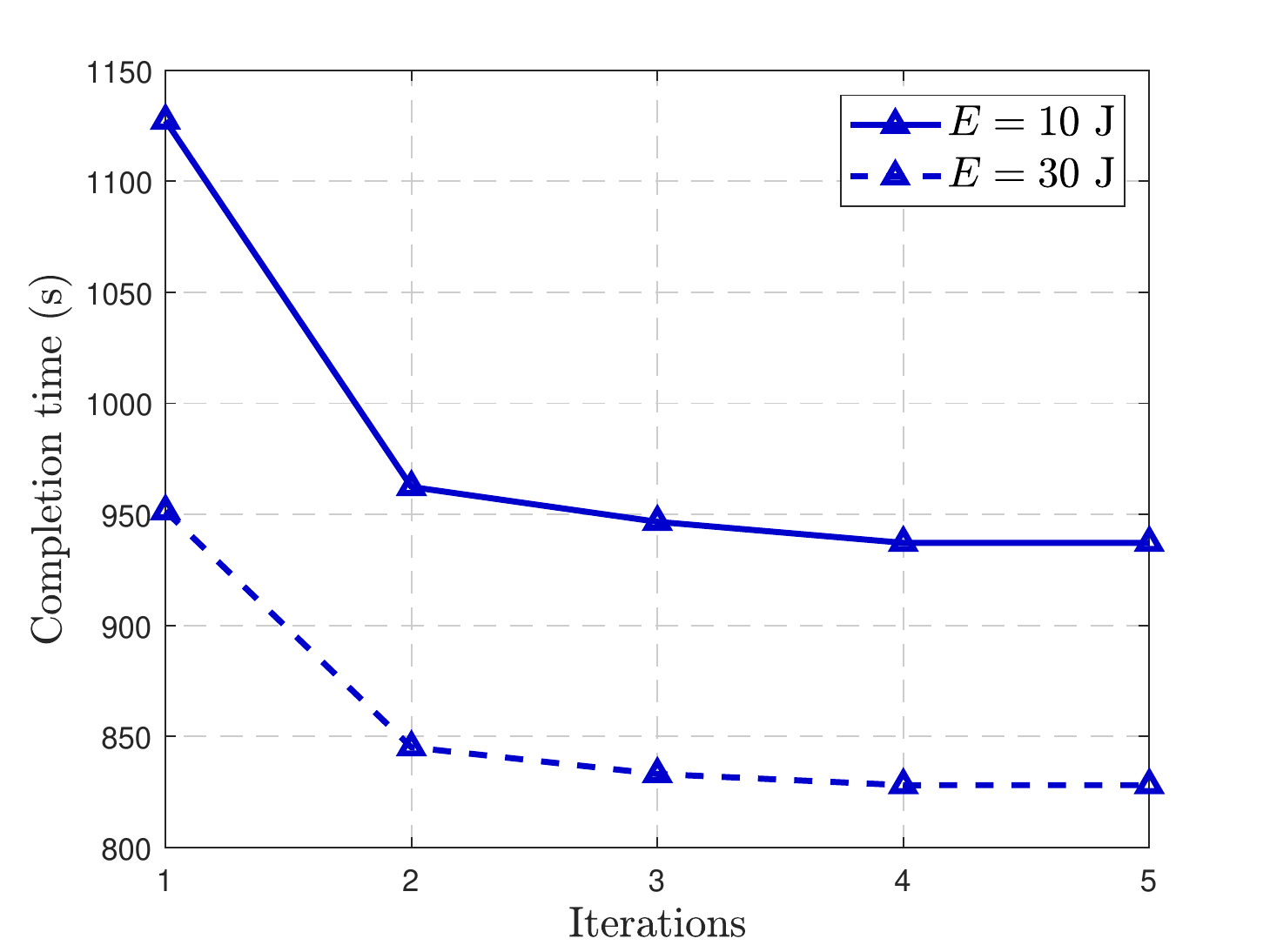}
		\vspace{-10mm}
		\caption{Convergence behaviors of Algorithm \ref{Algo:main} for different energy budgets of device.}
		\label{Fig:Iter}
	\end{minipage}
	\hspace{1mm}
	\begin{minipage}{.3\textwidth}
		\centering
		\includegraphics[width=1.1\textwidth]{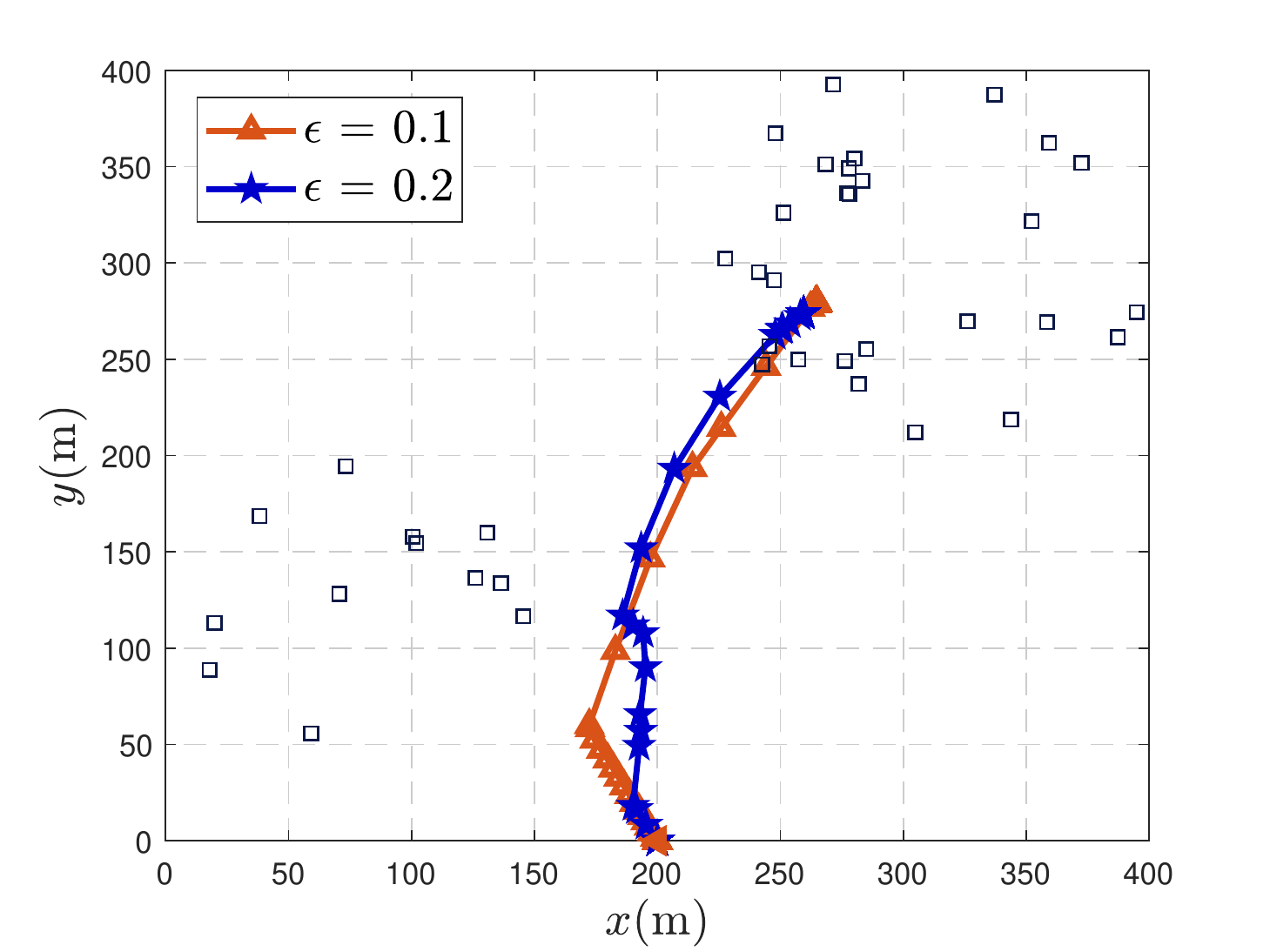}	
		\vspace{-10mm}
		\caption{UAV trajectories for different convergence accuracy of FL constraints.}
		\label{Fig:Traj}
	\end{minipage}
	\hspace{1mm}
	\begin{minipage}{.3\textwidth}
		\centering
		\includegraphics[width=1.1\textwidth]{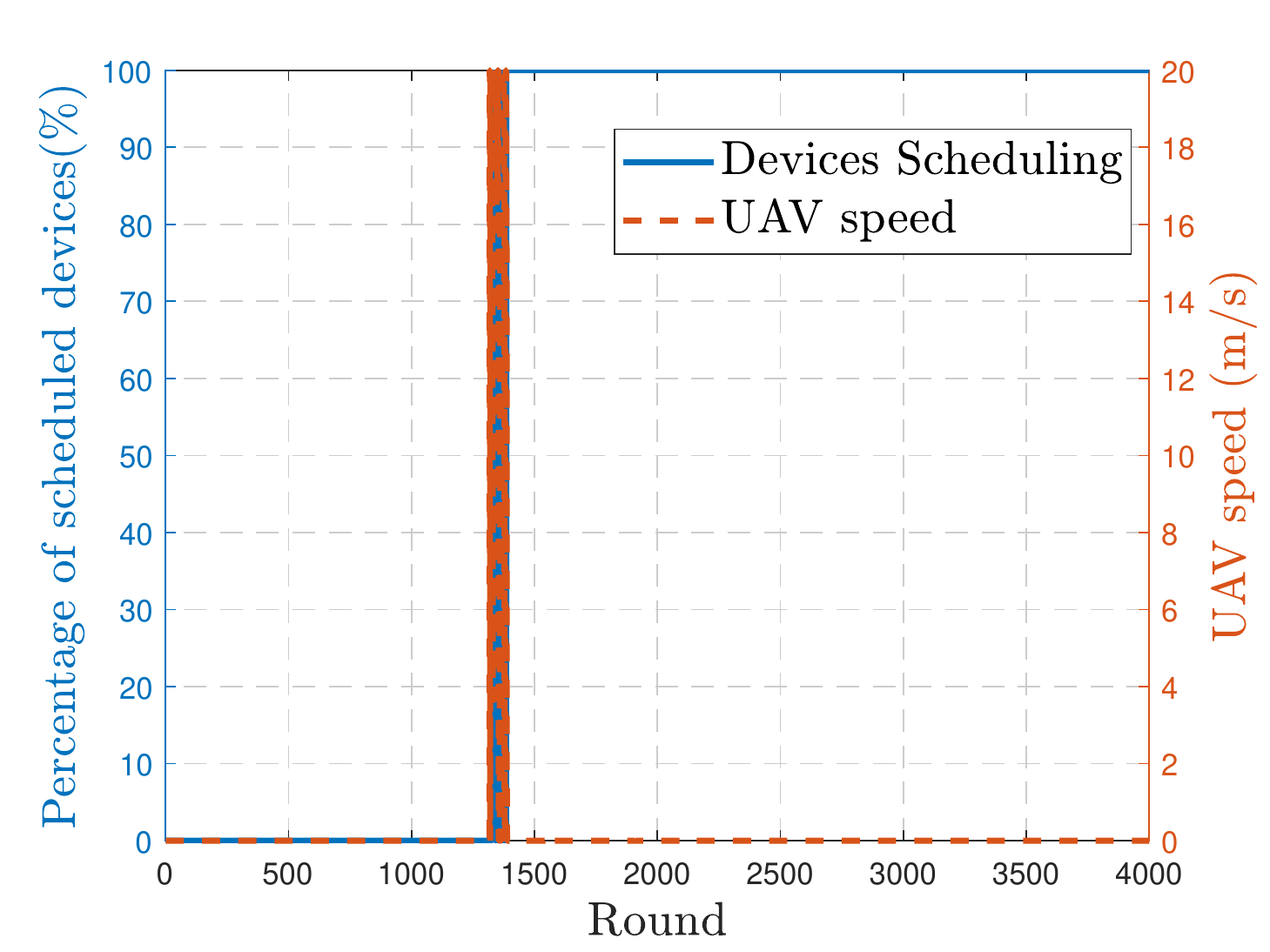}		
		\vspace{-10mm}
		\caption{The corresponding UAV speed and the percentage of scheduled devices over round when $\epsilon = 0.2$.}
		\label{Fig:speed}
	\end{minipage}
\end{figure}
\subsection{Optimized UAV Trajectory}
Fig. \ref{Fig:Iter} shows the convergence behaviors of the proposed BCD-based algorithm (i.e., Algorithm \ref{Algo:main}) with different device energy budgets $ E_k $, $\forall k$ when $\epsilon = 0.2$. 
All devices have the same energy budgets.
It is observed that for both schemes, the completion time with Algorithm \ref{Algo:main} decreases quickly with the number of iterations, and the two algorithms converge within 5 iterations with prescribed accuracy $\varepsilon= 10^{-3}$.
Besides, the proposed joint design achieves a smaller completion time under $E_k = 30$ J compared to $E_k = 10$ J.

Fig. \ref{Fig:Traj} illustrates the trajectories obtained by Algorithm \ref{Algo:main} under different convergence accuracy of FL $\epsilon$. 
Each trajectory is sampled every 5 rounds and the sampled points are marked with $\triangle$ by using the same colors as their corresponding trajectories. 
The devices' locations are marked by a dark blue $\square$. The predetermined locations of the UAV are red $\star$.
We can observe that with different $\epsilon$, the optimized trajectories have some similar behaviors.
First, the UAV visits the devices in both clusters as close as possible along an arc path, which means the UAV proactively shortens the devices' communication distance and thus schedules different devices to participate in FL training.
We also observe that the densities of sampling points on the trajectory in different time slots are different.
As a result of varying channel coefficients, the selected devices change during training.

To unveil the phenomena in trajectories obtained by these algorithms and the reasons for such phenomena, we plot the corresponding UAV speed and the percentage of scheduled devices over communication rounds, as shown in Fig. \ref{Fig:speed}. 	
It can be observed that there are three stages for behaviors of scheduled devices.
Furthermore, the behaviors of the corresponding UAV speed are consistent with that of the percentage of scheduled devices.
In the following, we elaborate on these three stages.
\begin{itemize}
	\item Stage I: As shown in Fig. \ref{Fig:speed}, from round 1 to about 1300 rounds, the percentage of participating devices per round is equal to 0, which means that no device is involved in FL training during this time interval.
	This is because in this interval when the UAV stays in the initial position, the communication straggler issue is severe.
	
	\item 
	Stage II: From about 1300 to 1400 rounds, the percentage of participating devices is dynamically changing, which means our proposed design outputs different device selection decisions when the channels change.

	\item Stage III: From about 1400 rounds to 4000 rounds, the percentage of participating devices equals $100 \%$ in each round, which means with the synergy of trajectory and energy, the optimized device scheduling not only satisfies the convergence accuracy but also maximizes data exploitation, thereby reducing model aggregation loss and accelerating FL convergence.
\end{itemize}
Importantly note that the effective communication rounds of FL in the proposed scheme only need about 2700 rounds to satisfy the convergence accuracy condition, i.e., $\epsilon = 0.2$.


\subsection{Performance Comparison for Completion Time }

\begin{figure}[t]
	\centering
	\begin{minipage}{.48\textwidth}
		\centering
	\includegraphics[width=8cm,height=6cm]{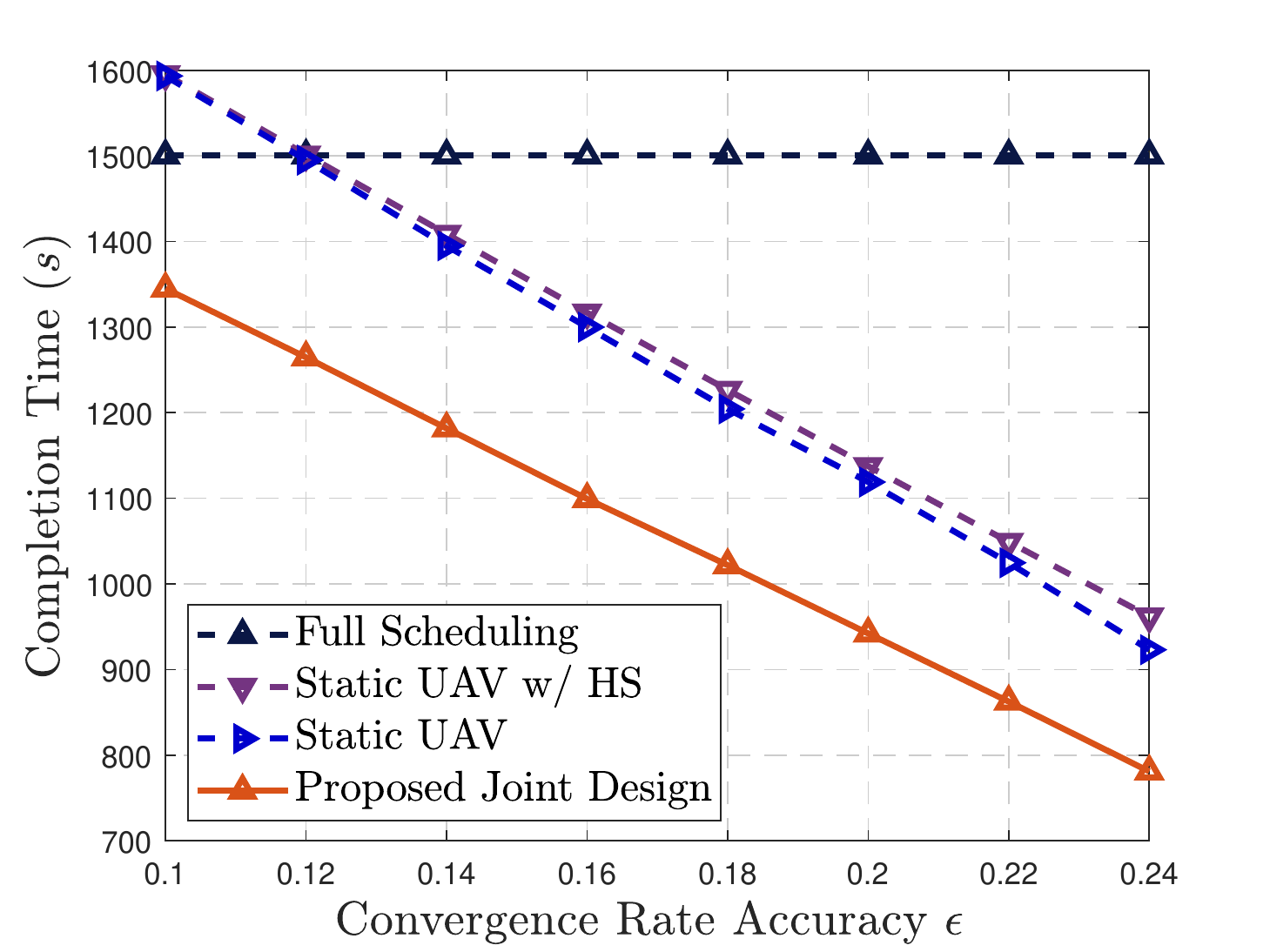}
		\vspace{-10mm}
		\caption{Completion time versus convergence accuracy of FL.}
		\label{Fig:Time_epsilon}
	\end{minipage}
	\hspace{4mm}
	\begin{minipage}{.48\textwidth}
		\centering
		\includegraphics[width=8cm,height=6cm]{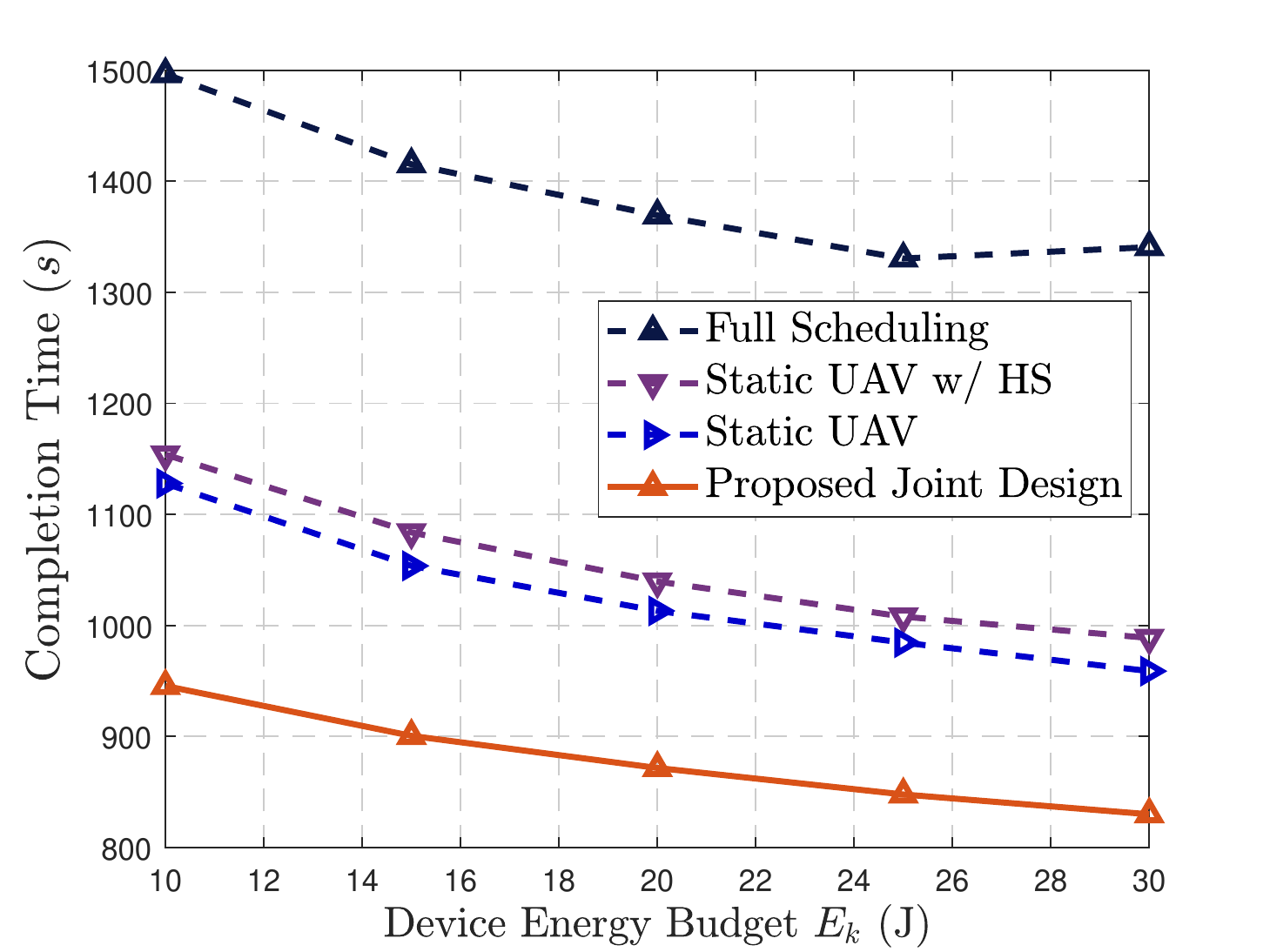}		
		\vspace{-10mm}
			\caption{Completion time versus device energy budget.}
		\label{Fig:Energy_time}
	\end{minipage}
\end{figure}

%


Fig. \ref{Fig:Time_epsilon} shows the effect of the convergence accuracy of FL on the completion time when $E = 10$ J.
Since the full scheduling scheme does not result in the model aggregation loss, its performance is invariant to the convergence accuracy $\epsilon$. 
However,  this scheme costs prohibitive completion time without device scheduling.
Note that this scheme may be forbidden in general since UAVs usually have a limited endurance due to practical physical constraints (e.g., 30 minutes for the typical rotary-wing UAV \cite{Wu2018MultiUAV}).
For the other three schemes with device scheduling, the completion time decreases with $\epsilon$, which is consistent with Theorem \ref{theorem:ProblemOriginal}.
In addition, the Static UAV scheme has a shorter completion time than the Static UAV w/ HS scheme since the former avoids the aggregation error explosion by excluding the weak devices.
However, in the case of strong straggler issues, the performance gain brought by device scheduling optimization alone is very limited.
Fortunately, the proposed joint design scheme incorporated trajectory design with device scheduling significantly reduce the completion time compared to all benchmarks.
This is because, the mobile UAV can shorten communication distances between the scheduled devices and always achieve a high transmission rate for local model uploading, thereby accelerating the model aggregation process.	
These results demonstrate in UAV-enabled systems, joint trajectory and device scheduling device is critical to mitigating the communication stragglers' effect.

Fig. \ref{Fig:Energy_time} shows the total mission completion time versus the device energy budget when $\epsilon = 0.2$.
It is observed that, for all schemes, the total completion time decreases as $E$ increases.
This is because, as $E$ increases, the scheduled devices can afford more energy to upload local models with higher transmission rates.
We also find that the proposed joint design achieves the smallest completion time compared to all benchmarks.
This observation further verifies the superiority of the proposed scheme to mitigate the stragglers' effect.


\begin{figure}[t]
	\centering
	\begin{minipage}{.48\textwidth}
		\centering
		\includegraphics[width=8cm,height=6cm]{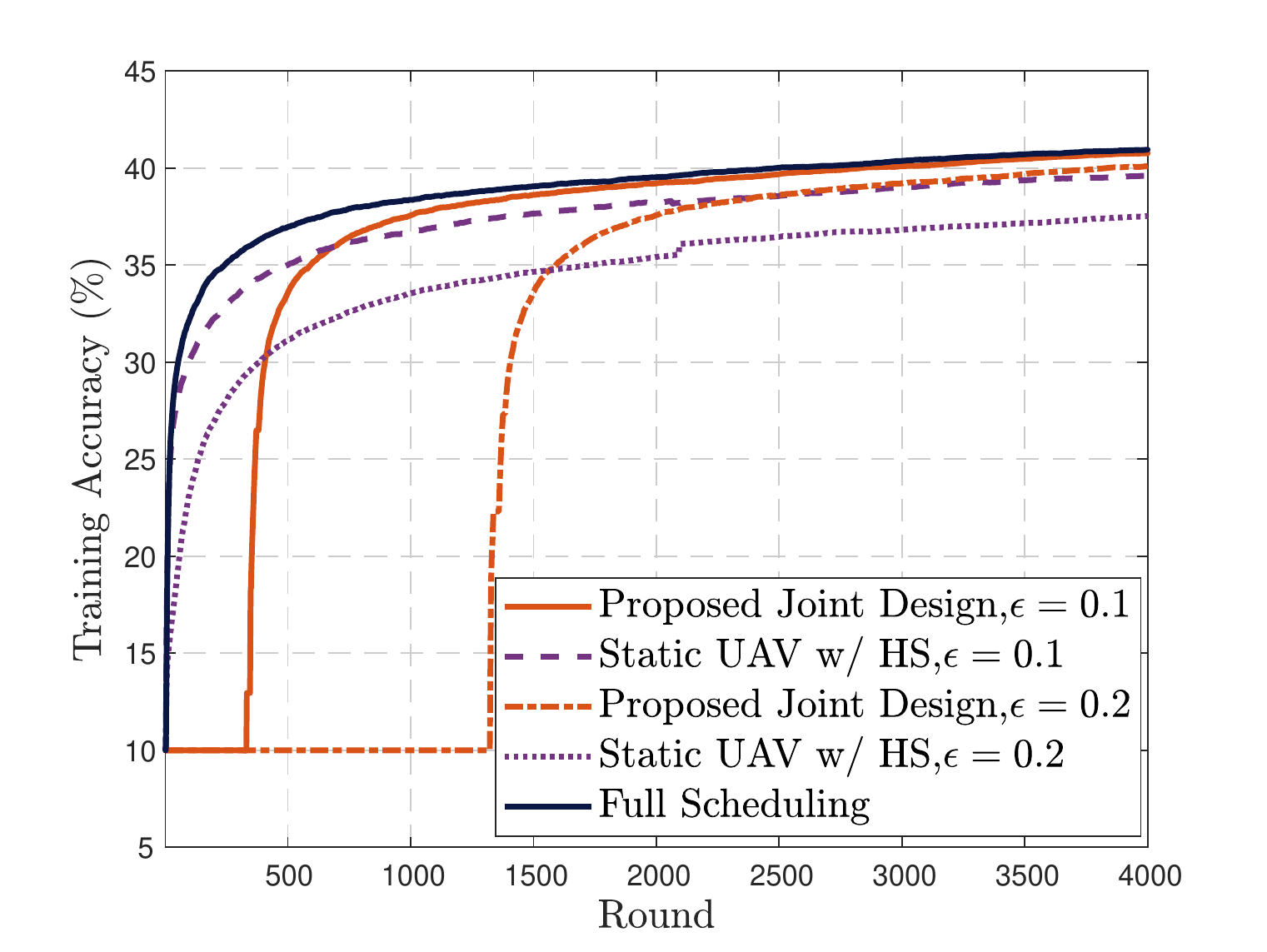}
		\vspace{-10mm}
		\caption{Training accuracy over communication rounds.}
		\label{Fig:Train_acc}	
	\end{minipage}
	\hspace{4mm}
	\begin{minipage}{.48\textwidth}
		\centering
		\includegraphics[width=8cm,height=6cm]{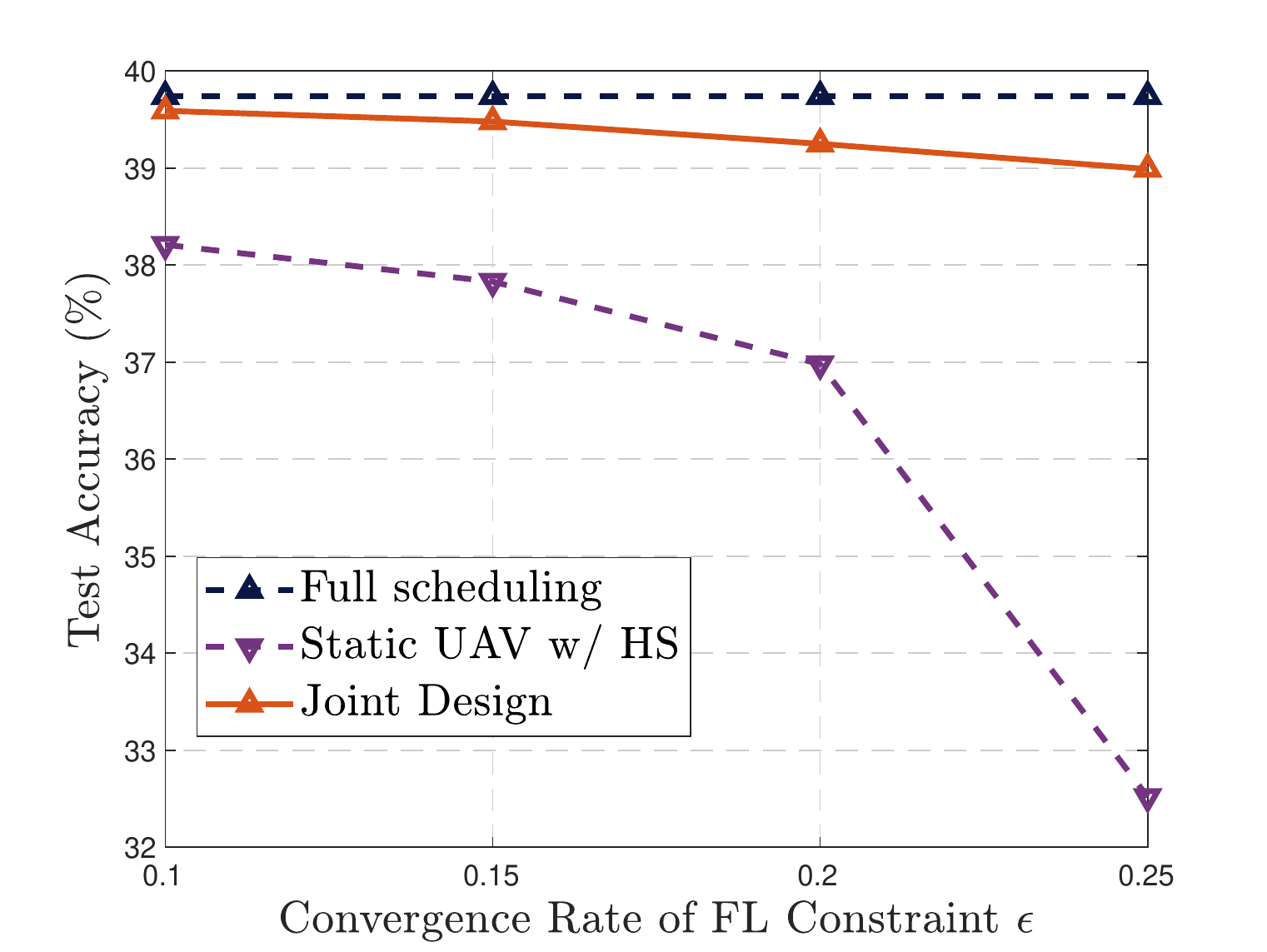}
		\vspace{-10mm}
		\caption{Test accuracy versus convergence accuracy of FL.}
		\label{Fig:Test_acc}
	\end{minipage}
\end{figure}

\subsection{Performance Comparison for FL}
To demonstrate the performance of our proposed joint design for dealing with FL tasks, we train image classifier models on the widely-used CIFAR-10 dataset.
In the simulations, the full scheduling scheme allows all devices to participate in the training process at each FL round that can achieve the most optimal FL performance, which serves as the lower bound.

Fig. \ref{Fig:Train_acc} shows training accuracy over communication rounds with different schemes and different $\epsilon$.
We can observe that the convergence accuracy $\epsilon$ heavily affects the learning performance of the Static UAV w/ HS since a larger value of $\epsilon$ implies that the data are less exploited in this scheme due to the fact of device exclusion.
In contrast, the proposed joint design still achieves a training accuracy that is close to the ideal benchmark even as $\epsilon$ increases.
For example, we can see from Fig. \ref{Fig:speed} with $\epsilon = 0.2$, the percentage of participating devices is $100 \%$ from about 1400 rounds to 4000 rounds.
This is because the proposed joint design scheme overcomes the straggler issue and maximizes data exploitation.
Specifically, the UAV's trajectory tends to be close to the devices to establish strong communication links, thus maximizing data exploitation.
In contrast, at some UAV locations (e.g., the initial position), where the communication straggler issue is severe, the device schedule is optimized to let no device be involved in FL training.
This is the reason why the curve of the proposed joint design scheme is invariant during the initial communication rounds.
On the other hand, by properly designing the trajectory of the UAV, short-distance LoS links can be proactively and dynamically established for any device, which ensures all devices have equal opportunity to participate in FL training instead of device exclusion, thereby increasing the data exploitation and training accuracy.



Fig. \ref{Fig:Test_acc} shows the test accuracy versus convergence rate of FL constraints with different schemes when $E = 10$ J.
It is observed that for all schemes except for the full scheduling scheme, the test accuracy decreases as $\epsilon$ increases.
Observed from Fig. \ref{Fig:Time_epsilon} and Fig. \ref{Fig:Test_acc}, it shows that there exists a fundamental trade-off between the mission completion time and the learning performance of FL.
However, by exploiting the mobility of the UAV and flexibility of device scheduling, our proposed joint design significantly improves the trade-off compared to the static UAV w/ HS scheme.
For example, when $\epsilon = 0.2$, the completion time of FL achieved by our proposed joint design, static UAV w/HS, and the full scheduling scheme are respectively $930$ s, and $1150$ s, $1500$ s.
Meanwhile, when $\epsilon = 0.2$, the test accuracy achieved by our proposed joint design, static UAV w/HS, and the full scheduling scheme are respectively $38.99 \%$ s, and $32.52 \%$ s, and $39.74 \%$.
For static UAV w/HS, compared to the full scheduling scheme, the completion time reduces $23.33\%$ while the test accuracy reduces $7.22\%$.
Fortunately, for our proposed joint design, compared to the full scheduling scheme, the completion time reduces $38 \%$ while the test accuracy only reduces $0.71 \%$.

\section{Conclusion}\label{Section:conclusion}
In this paper, we studied the UAV-assisted FL system to effectively address the straggler issue by exploiting the UAV's high altitude and mobility when the terrestrial PS was unavailable.
Specifically, a mobile UAV was deployed as a flying PS to exchange model parameters from distributed ground devices and train a shared FL model.
We focused on the joint consideration of device scheduling, UAV trajectory, and time allocation to minimize the completion time required for FL to converge to the desired accuracy level.
Before proceeding with problem formulation, we first theoretically analyzed the effect of device scheduling on the convergence rate of FL without the assumptions of convexity. 
We also provided a convergence bound for the average norm of the global gradient.
Based on such convergence results, we formulated the completion time minimization problem, taking into account the practical device's energy budget, UAV's mobility constraints, and convergence accuracy constraint.
Although the formulated problem was non-convex, we exploited its structures to decompose it into two sub-problems, followed by deriving closed-form solutions via the Lagrange dual ascent method.
Simulation results demonstrated that our proposed joint design can significantly reduce completion time and enhance the tradeoff between completion time and prediction accuracy compared to the existing benchmarks.

\appendix 
\subsection{Proof of Theorem \ref{theorem:convergence result}}\label{theorem:convergence result proof}
	Let ${\bm g}_{n} \triangleq \frac{\sum_{k=1}^Ka_k[n+1]D_k\nabla \mathcal{F}_k(\bm w_{n})}{\sum_{k=1}^Ka_k[n+1]D_k}$.
	Combining the update rules at devices in \eqref{Eq:local update} and the update rule at the server in \eqref{Eq:global update}, the update of the global FL model at round $n+1$ is given by
	$	\bm w_{n+1} = \bm w_{n} - \eta \frac{\sum_{k=1}^Ka_k[n+1]D_k\nabla \mathcal{F}_k(\bm w_{n})}{\sum_{k=1}^Ka_k[n+1]D_k}
		= \bm w_{n} - \eta {\bm g}_{n}.$
	From the assumption in \eqref{Neq:Smooth}, we have that
	\begin{eqnarray}
	\mathcal{F}(\bm w_{n+1})- \mathcal{F}(\bm w_{n})
	&&	\leq  \langle\nabla \mathcal{F}(\bm w_{n}), \bm w_{n+1} - \bm w_{n}\rangle + (L/2)\|\bm w_{n+1}-\bm w_n\|_2^2, \nonumber \\
	&&= \langle\nabla \mathcal{F}(\bm w_{n}), -\eta\bm {g}_n \rangle + (\eta^2L/2) \|{\bm g}_n\|_2^2.
	\end{eqnarray}
	Using $\bm g_n = \nabla \mathcal{F}(\bm w_{n}) -(\nabla \mathcal{F}(\bm w_{n})-\bm g_n) $, we derive that
	\begin{eqnarray}
	\mathcal{F}(\bm w_{n+1})- \mathcal{F}(\bm w_{n}) &\leq & 
		(\eta L/2-1)\eta\|\nabla \mathcal{F}(\bm w_{n})\|_2^2 
		 + (\eta^2 L/2)\|\nabla \mathcal{F}(\bm w_{n})-\bm g_n\|_2^2 \nonumber \\
		&+& (1-\eta L)\eta\langle\nabla \mathcal{F}(\bm w_{n}), \nabla 
		 \mathcal{F}(\bm w_{n})-\bm {g}_n \rangle. \label{Neq:loss diff} 
	\end{eqnarray}
	The third term at the right-hand side in \eqref{Neq:loss diff} can be upper bounded as
	\begin{eqnarray}
		 \langle\nabla \mathcal{F}(\bm w_{n}), \nabla \mathcal{F}(\bm w_{n})-\bm {g}_n \rangle &\leq& \frac{1}{2}(\|\nabla \mathcal{F}(\bm w_{n})\|_2^2  
		+ \|\nabla \mathcal{F}(\bm w_{n})-\bm {g}_n \|_2^2).\label{Neq:inner product}
	\end{eqnarray}
	Given $0<\eta \leq \frac{1}{L}$, we substitute \eqref{Neq:inner product} into \eqref{Neq:loss diff} and 
	rearrange the result, yielding
	\begin{eqnarray} \label{Neq:11}
		\|\nabla \mathcal{F}(\bm w_{n})\|_2^2  
		\leq \frac{2}{\eta} (\mathcal{F}(\bm w_{n})-\mathcal{F}(\bm w_{n+1})) + \|\nabla \mathcal{F}(\bm w_{n})-\bm g_n\|_2^2. 
	\end{eqnarray}
Therein, the second term , namely the global gradient deviation error in expression \eqref{Neq:11} is can be rewritten as 
\begin{eqnarray} \label{Neq:error1}
	&&\!\!\!\!\!\!	\|\nabla \mathcal{F}(\bm w_{n})-\bm g_n\|_2^2
	=   \Big\|\frac{\sum_{k=1}^Ka_k[n+1]D_k\nabla \mathcal{F}_k(\bm w_{n})}{\sum_{k=1}^Ka_k[n+1]D_k}  -\frac{1}{D}\sum_{k=1}^KD_k\nabla \mathcal{F}_k(\bm w_{n})\Big\|_2^2 \nonumber \\
\!\!\!\!	= && \!\!\!\! \Big(\Big\|\frac{\sum_{k=1}^K\!D_k(1\!-\!a_k[n\!+\!1])}{D\!\sum_{k=1}^K\!a_k[n\!+\!1]D_k}\sum_{k=1}^Ka_k[n\!+\!1]D_k\nabla \mathcal{F}_k(\bm w_{n})\|
	\!+\!\|\frac{1}{D}\!\sum_{k=1}^K(1-a_k[n+1])D_k\nabla \mathcal{F}_k(\bm w_{n})\Big\|\Big)^2. \nonumber\\
\end{eqnarray}	
Thus, based on \eqref{Neq:error1} and triangle inequality, the global gradient deviation error is bounded as
	\begin{eqnarray} \label{Neq:error}
	\!\!\!\!	\|\nabla \mathcal{F}(\bm w_{n})-\bm g_n\|_2^2
		 && \leq   \Big(\!\frac{\sum_{k=1}^K\!D_k(1\!-\!a_k[n\!+\!1])}{D\sum_{k=1}^Ka_k[n\!+\!1]D_k}\!\sum_{k=1}^K\!a_k[n\!+\!1]\!\sum_{i = 1}^{D_k}\! \|\nabla\! f(\bm w_n;\bm x_{ki}, y_{ki})\|_2\nonumber\\
		&& + \frac{1}{D}\sum_{k=1}^K(1-a_k[n+1])\sum_{i=1}^{D_k} \|\nabla f(\bm w_n;\bm x_{ki}, y_{ki})\|_2\Big)^2\nonumber\\
		 \leq &&\!\!\!\! \frac{4}{D^2}(\sum_{k=1}^K(1-a_k[n+1])D_k)^2\kappa
		\leq  \frac{4K\kappa }{D^2}\sum_{k=1}^K(1-a_k[n+1])D_k^2.	
	\end{eqnarray}	
Therefore, based on \eqref{Neq:11} and \eqref{Neq:error}, 
	it follows that
	\begin{eqnarray}
	\frac{1}{N}\sum_{n=0}^{N-1}\|\nabla \mathcal{F}(\bm w_{n})\|^2 
		\leq&&\!\!\!\! \frac{2}{N\eta}(\mathcal{F}(\bm w_0)- F(\bm w_{N})) + \frac{1}{N}\sum_{n=0}^{N-1}\|\nabla \mathcal{F}(\bm w_{n})-\bm g_n\|^2\nonumber\\
		\leq&&\!\!\!\! \frac{2}{N\eta}(\mathcal{F}(\bm w_0)- F^{\star}) + \frac{1}{N}\sum_{n=0}^{N-1}\|\nabla \mathcal{F}(\bm w_{n})-\bm g_n\|^2 \nonumber\\
	\overset{\eqref{Neq:error}}	\leq&&\!\!\!\! \frac{2}{N\eta}(\mathcal{F}(\bm w_0)- F^{\star}) +  \frac{4K\kappa}{ND^2}\sum_{n=0}^{N-1}\sum_{k=1}^K(1-a_k[n+1])D_k^2. 
	\end{eqnarray}
	This completes the proof.

\vspace{-5mm}
\subsection{Proof of Proposition \ref{proposition:ProblemFeasibility}} \label{Proposition:ProblemFeasibility proof}
	Given an arbitrary convergence threshold $\epsilon$, to make constraint \eqref{Cons:F accuracy} in problem \eqref{Problem: original} feasible, if and only if the minimum communication rounds	satisfies
	$
		N \geq \lceil\frac{2(\mathcal{F}(\bm w_{0})-F^{\star})}{\epsilon \eta} \rceil,
	$
	where $a_k[n] = 1, \forall k, \forall N$.	
	In addition, we have that
	\begin{eqnarray}\label{Neq:1}	
		&&\frac{4K\kappa}{ND^2}\sum_{n=0}^{N-1}\sum_{k=1}^K(1-a_k[n+1])D_k^2 
		\leq\!\frac{4K\kappa}{ND^2}\max_kD_k^2\sum_{n=0}^{N-1}\sum_{k=1}^K(1\!-\!a_k[n\!+\!1]). 
	\end{eqnarray}	
By instituting the right term of  inequality \eqref{Neq:1} into  constraint \eqref{Cons:F accuracy},	if  \eqref{Cons:F accuracy} is satisfied, then
	\begin{eqnarray}\label{Neq: scheduled devices}	
		\sum_{n=0}^{N-1}\sum_{k=1}^Ka_k[n+1] \geq	|\mathcal{S}|,
	\end{eqnarray}		
where $	|\mathcal{S}| =\lceil KN-(\epsilon-\frac{2}{N\eta}(\mathcal{F}(\bm w_0)-F^{\star}))\frac{ND^2}{4K\kappa \max_kD_k^2}\rceil$.
For device $k$, the minimum energy consumption for local model uploading is given by
	\begin{eqnarray}
	\!\!\!\!\!\!	\lim_{\tau_k[n] \longrightarrow +\infty}	E_k^{\text{comm}}[n] &=&  \lim_{\tau_k[n] \longrightarrow +\infty}\frac{\tau_k[n]\sigma^2}{h_k[n]}\Big(2^{\frac{a_k[n]s}{B\tau_k[n]}}-1\Big) 
		 =  \frac{ a_k[n]s\sigma^2\ln 2}{Bh_k[n]}
		\geq \frac{ a_k[n]s\sigma^2H^2\ln 2}{B\beta_0}.
	\end{eqnarray}		
Besides, if constraints \eqref{Cons:E budget} are satisfied, we have that
	\begin{eqnarray}\label{Neq: 21}
		\frac{s\sigma^2H^2\ln 2}{B\beta_0}\sum_{k=1}^{K}\sum_{n=1}^{N} a_k[n]&\leq& \sum_{k=1}^{K}\sum_{n=1}^{N}\Big(E_k^{\text{comm}}[n] +   E_{k}^{\text{comp}}[n]\Big) 
		 \leq \sum_{k=1}^{K}{E}_{k}.
	\end{eqnarray}	
	Combining \eqref{Neq: scheduled devices} with \eqref{Neq: 21}, to make constraints \eqref{Cons:E budget} and constraint \eqref{Cons:F accuracy} feasible, we have that
	\begin{eqnarray}\label{Neq: 2}
	&&	\frac{s\sigma^2H^2\ln 2}{B\beta_0} 	|\mathcal{S}|  \leq  \sum_{k=1}^{K}{E}_{k}, 
		N \geq \lceil\frac{2(\mathcal{F}(\bm w_{0})-F^{\star})}{\epsilon \eta} \rceil.
	\end{eqnarray}		
Thus, this completes the proof.

\vspace{-2mm}
\subsection{Proof of Theorem \ref{theorem:ProblemOriginal}} \label{theorem:ProblemOriginal proof}
	Denote the optimal solutions of problem \eqref{Problem: original} with $\epsilon^{\star}$	and $\hat{\epsilon}$ by with $\mathcal{E}^{\star} = \{\bm A^{\star}, \bm \Gamma^{\star}, \bm Q^{\star}, \bm \delta^{\star}\}$ and $\hat{\mathcal{E}} = \{ \bm {\hat {A}}, \bm {\hat{\Gamma}}, \bm {\hat{Q}} , \bm {\hat{ \delta}}\}$, respectively. 
	To prove Theorem \ref{theorem:ProblemOriginal}, we only need to show that $\sum \delta^{\star}[n] \geq \sum \hat{ \delta}[n]$ holds when $\epsilon^{\star} \leq \hat{\epsilon}$. Note that in problem \eqref{Problem: original},  the convergence threshold is only involved in constraint \eqref{Cons:F accuracy}. 
	Thus, we have the following inequalities
	\begin{eqnarray}
		&&\frac{2}{N\eta}(\mathcal{F}(\bm w_{0})-F^{\star})
		+\frac{4K\kappa}{ND^2}\sum_{n=0}^{N-1}\sum_{k=1}^K(1-a_k^{\star}[n+1])D_k^2 \leq \epsilon^{\star}	\leq \hat{\epsilon},
	\end{eqnarray}	
	which implies that $\mathcal{E}^{\star}$ is also a feasible solution of problem  \eqref{Problem: original} with $\hat{\epsilon}$. Since $\sum{\hat{ \delta}}[n]$ is the minimum objective value	of problem \eqref{Problem: original} with $\hat{\epsilon}$, it follows that $\sum \delta^{\star}[n] \geq \sum \hat{ \delta}[n]$, which thus completes the proof of Theorem \ref{theorem:ProblemOriginal}.

\subsection{Proof of Theorem \ref{theorem:tau-solution}} \label{theorem:tau-solution proof}
Let the first derivative of the objective function of problem \eqref{Subproblem:Primal A and tau} w.r.t. $\tau_k[n]$ be zeros, i.e.,
\begin{eqnarray}\label{Eq:derivative-tau}
\!\!\!\!\sum_{k=1}^{K} \!\mu_{k}[n]  \!+\! \lambda_k \frac{ \sigma^2}{h_k[n]}\Big(\!2^{\frac{a_{k}[n]s}{B\tau_k[n]}}\!-\!\frac{a_{k}[n]s\ln2}{B\tau_k[n]}2^{\frac{a_{k}[n]s}{B\tau_k[n]}}\!-\!1\!\Big) \!= \!0.
\end{eqnarray}
We can rewritten  equation \eqref{Eq:derivative-tau} as
\begin{eqnarray}
\Big(\frac{a_{k}[n]s\ln2}{B\tau_k[n]}-1\Big) 2^{\frac{a_{k}[n]s}{B\tau_k[n]}}=    \frac{h_k[n]\sum_{k=1}^{K} \mu_{k}[n]} { \lambda_k \sigma^2}-1.
\end{eqnarray}
By solving the above equation, one can have
$
	\frac{a_{k}[n]s\ln2}{B\tau_k[n]}-1 =  \mathcal{W}\Big( \frac{h_k[n]\sum_{k=1}^{K} \mu_{k}[n]} { e\lambda_k \sigma^2}-\frac{1}{e}\Big).
$
By simplifying the above equation, we can have that
\begin{eqnarray}
\tau_{k}[n] = \frac{a_{k}[n]s\ln2}{ B\Big(\mathcal{W}\Big( \frac{h_k[n]\sum_{k=1}^{K} \mu_{k}[n]} { e\lambda_k \sigma^2}-\frac{1}{e}\Big)+1\Big)}.
\end{eqnarray}
Thus, this completes the proof.

\subsection{Proof of Proposition \ref{convergence proposition}}\label{convergence proposition proof}
We denote  $T^{\star} (\bm A^{\star}, \bm \Gamma^{\star},  \bm \delta^{\star}, \bm Q^{\star}) =  \sum_{n=1}^N\delta^{\star}[n]$  as the objective value of problem \eqref{Problem: original} with solution  $\{\bm A^{\star}, \bm \Gamma^{\star},  \bm \delta^{\star}, \bm Q^{\star}\}$.
As shown in step 5 of Algorithm \ref{Algo:main}, a feasible solution of problem \eqref{Problem:Q-feasibility} (i.e., $\{ \bm A^{r}, \bm \Gamma^{r},  \bm \delta^{r}, \bm Q^{r} \}$) is also feasible to problem \eqref{Subproblem:scheduling and time}.
We denote $\{ \bm A^{r}, \bm \Gamma^{r},  \bm \delta^{r}, \bm Q^{r} \}$ and $\{ \bm A^{r +1}, \bm \Gamma^{r+1},  \bm \delta^{r+1}, \bm Q^{r+1} \}$ as a  feasible solution of   problem \eqref{Problem: original} at the $r$-th and $(r+1)$-th iterations, respectively. 

Because, given $\bm Q^{r}$ as  shown in step 4 of Algorithm \ref{Algo:main}, $\{ \bm A^{r +1}, \bm \Gamma^{r+1},  \bm \delta^{r+1}\}$ is the optimal solution to problem \eqref{Subproblem:scheduling and time}, we have
\begin{eqnarray} \label{p1}
	T(\bm A^{r}, \bm \Gamma^{r},  \bm \delta^{r}, \bm Q^{r})  \geq T (\bm A^{r+1}, \bm \Gamma^{r+1},  \bm \delta^{r+1}, \bm Q^{r}) . 	
\end{eqnarray}
Similarly, given $\bm A^{r+1}, \bm \Gamma^{r+1},  \bm \delta^{r+1}$ as  shown in step 5 of Algorithm  \ref{Algo:main}, $\bm Q^{r+1}$ is the optimal solution to problem \eqref{Subproblem:Q-objective}, it follows that
\begin{eqnarray}\label{p3}
\!\!\!\!T (\bm A^{r+1}, \bm \Gamma^{r+1},  \bm \delta^{r+1}, \bm Q^{r}) \!=\! T (\bm A^{r+1}, \bm \Gamma^{r+1},  \bm \delta^{r+1}, \bm Q^{r+1}). 	
\end{eqnarray}
This holds because the original objective function $T$ is independent of $\bm Q$ but depends on $\bm \delta$.
Based on $\eqref{p1}$ and $\eqref{p3}$, we further obtain 
\begin{eqnarray}
	T (\bm A^{r+1}, \bm \Gamma^{r+1},  \bm \delta^{r+1}, \bm Q^{r+1}) \leq T (\bm A^{r}, \bm \Gamma^{r},  \bm \delta^{r}, \bm Q^{r}), 	
\end{eqnarray}
which shows that the objective value of problem \eqref{Problem: original} is always decreasing over iterations. 
Therefore, the proposed BCD-LD algorithm converges. This thus completes the proof.

\bibliography{Reference} 
\bibliographystyle{ieeetr}

\end{document}